\documentclass[sigconf]{acmart}
\usepackage{xcolor}
\usepackage{tabularx} 
\usepackage{multirow} 

\definecolor{gray1}{gray}{0.7}
\definecolor{gray2}{gray}{0.5}
\definecolor{gray3}{gray}{0.3}
\definecolor{textgray}{gray}{0.2}
\newcommand{\gain}[1]{\colorbox{green}{\textcolor{black}{#1}}}
\newcommand{\drop}[1]{\colorbox{pink}{\textcolor{black}{#1}}}

\AtBeginDocument{%
  }

\setcopyright{acmlicensed}
\copyrightyear{2018}
\acmYear{2018}
\acmDOI{XXXXXXX.XXXXXXX}
\acmConference[Conference acronym 'XX]{Make sure to enter the correct
  conference title from your rights confirmation email}{June 03--05,
  2018}{Woodstock, NY}
\acmISBN{978-1-4503-XXXX-X/2018/06}

\begin{document}

\title{Birds of a Feather Cluster Nearby: a Proximity‑Aware Geo-Codebook for Local Service Recommendation}


\author{Tian He}
\authornote{Both authors contributed equally to this research. This work was completed by the two authors during their internship at Meituan.}
\email{hetian@bupt.edu.cn}
\affiliation{%
  \institution{Beijing University of Posts and Telecommunications; Meituan}
  \city{Beijing}
  \country{China}
}

\author{Chen Yang}
\authornotemark[1]
\email{yangchen25@bit.edu.cn}
\affiliation{%
  \institution{Beijing Institute of Technology; Meituan}
   \city{Beijing}
  \country{China}}

\author{Jiawei Zhang}
\email{zhangjiawei43@meituan.com}
\affiliation{%
  \institution{Meituan}
  \city{Beijing}
  \country{China}}

\author{Lin Guo}
\email{guolin08@meituan.com}
\affiliation{%
  \institution{Meituan}
  \city{Beijing}
  \country{China}}

\author{Wei Lin}
\email{linwei31@meituan.com}
\affiliation{%
  \institution{Meituan}
  \city{Beijing}
  \country{China}}

\author{Zhuqing Jiang}
\authornote{Corresponding Author. provided 
academic guidance on the research content. }
\email{jiangzhuqing@bupt.edu.cn}
\affiliation{%
  \institution{Beijing University of Posts and Telecommunications}
  \city{Beijing}
  \country{China}
}

\renewcommand{\shortauthors}{Trovato et al.}

\begin{abstract}
Generative recommendation systems are increasingly adopted in local service platforms, where semantic relevance alone is insufficient without strict geographic feasibility. A key technical challenge lies in semantic ID (SID) tokenization, which directly impacts the recommendation performance. However, existing semantic codebooks neglect geographic constraints, often resulting in recommendations that are semantically relevant yet geographically unreachable. To address this limitation, we propose  \textbf{Pro-GEO}, a \textit{\textbf{Pro}ximity-aware \textbf{GEO}-codebook}. Pro-GEO establishes a geo-centroid local coordinate system to capture intra-cluster spatial relationships and a geo-rotary position encoding mechanism that models geographic proximity as orthogonal rotational transformations in the high-dimensional embedding. This design enables semantic and spatial signals to be jointly modeled in a balanced manner, without reducing geographic information to a weak auxiliary feature. Extensive experiments conducted on a large-scale industrial dataset reveal that Pro-GEO significantly outperforms state-of-the-art methods. In particular, Pro-GEO reduces the average geographic clustering distance by 45.60\% and achieves a 1.87\% improvement in Hit@50, highlighting its effectiveness for real-world local service recommendation.

\end{abstract}






\maketitle
\begin{figure}[h]
    \centering
    \includegraphics[width=\linewidth]{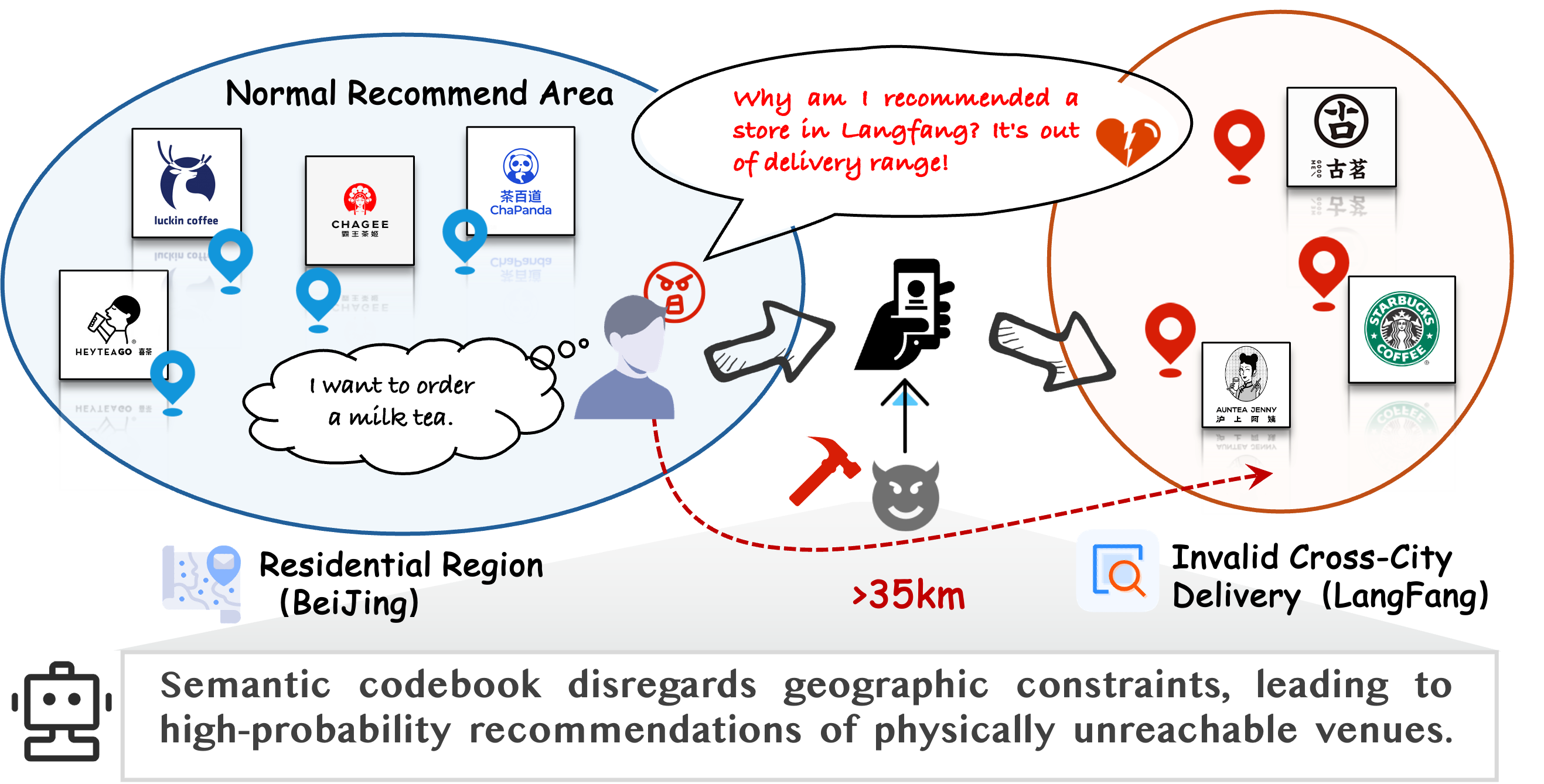}
    \vspace{0pt}
    \caption{Illustration of local lifestyle recommendation.}  
    \label{fig:1}
    \vspace{0pt}
\end{figure}
\section{Introduction}

Local service platforms such as Meituan, DoorDash and Yelp have become integral to daily life \cite{jiang2025plug,wang2025fim,hu2025dynamic}, where user satisfaction hinges not only on semantic relevance and personalization, but also on strict geographic feasibility. Recommending a semantically relevant yet geographically inaccessible Point-of-Interest (POI) can significantly degrade user experience and lead to inefficient resource allocation. Therefore, several local recommendation methods \cite{lian2014geomf,lim2020stp} explicitly incorporate spatial and mobility signals into geo-aware multi-stage ranking pipelines, such as user location, POI coordinates, and spatio-temporal patterns.

Recently, industrial recommender systems have increasingly shifted from multi-stage pipelines to end-to-end generative recommendation frameworks \cite{chen2025unisearch,hong2025eager,hou2025generating}. Instead of retrieving and ranking items from a predefined candidate set, these models directly generate item sequences using tokenized Semantic IDs (SIDs). To address the scalability challenges of large item vocabularies, SIDs encode each POI as a fixed-length code drawn from a shared codebook, effectively mitigating vocabulary explosion. Representative approaches such as TIGER \cite{rajput2023recommender} employed RQ-VAE to quantize residuals between latent embeddings and codebook centers, enabling coarse-to-fine semantic encoding. This paradigm was widely adopted by subsequent models, including LC-Rec \cite{zheng2024adapting}, COBRA \cite{yang2025sparse}, GFlowGR \cite{wang2025gflowgr}, and STREAM-Rec \cite{zhang2025slow}. More recently, RQ-KMeans \cite{han2026feature,deng2025onerec} was proposed to construct denser semantic codebooks. Despite their success, existing SID-based generative recommender systems primarily cluster POIs using high-variance semantic signals, such as categories or user interaction patterns, with little explicit consideration of geographic feasibility. As illustrated in Fig.~\ref{fig:1}, a hamburger restaurant in Langfang, Hebei Province, may be recommended to a user searching in Tongzhou District, Beijing. Although these locations are geographically adjacent, they are typically separated by more than 35 kilometers and fall under different municipal jurisdictions, rendering such recommendations infeasible for real-world local services. To mitigate this issue, recent studies have attempted to incorporate geographic information into codebook construction. GNPR-SID \cite{wang2025generative} attempted Google Maps Plus Codes to discretize latitude–longitude coordinates into regional identifiers, enabling information sharing among spatially proximate POIs. OneLoc \cite{wei2025oneloc} further integrated geographic context as a core feature in residual quantization to produce a geo-aware semantic codebook. While these methods improve spatial coherence to some extent, they rely on global geographic features and treat location primarily as an auxiliary attribute. Whether such global representations are sufficient for accurately modeling local spatial relationships in generative recommendation remains an open question.

Our large-scale analysis on a leading local-life service platform in China, following the method described in \cite{wang2025generative}, reveals that over 50\% of POIs sharing the same semantic ID are geographically separated by more than 40 kilometers. This indicates a fundamental structural bias in semantic codebook quantization: high-variance semantic features dominate the clustering objective, relegating geographic signals to weak regularizers that fail to meaningfully constrain spatial feasibility. To address this issue, we propose Pro-GEO, a proximity-aware GEO-codebook designed to explicitly encode relative geographic structure within semantic tokenization. Our work is guided by two key research questions:

\textit{\textbf{Q1: Are global geographic features (e.g., latitude–longitude or geohash) sufficient for local service recommendation}, or should spatial information be modeled in a relative, behavior-aware coordinate system?}

\textit{\textbf{Q2: How can semantic relevance and geographic proximity be jointly encoded in the codebook}, such that neither collapses into a weak regularizer of the other?}

For \textbf{Q1}, while latitude–longitude coordinates offer universal location, their direct application in representation learning fails to preserve meaningful local spatial relationships, especially in large-scale datasets where small coordinate differences can correspond to substantial physical distances. Motivated by this observation, we propose a \textbf{\textit{geo-centroid local coordinate system}} that redefines how spatial information is represented for POI recommendation. In the early stages of clustering, high-dimensional semantic vectors dominate, resulting in SID subsets that are semantically coherent. Within each such subset, we define a geographic centroid by aggregating the actual latitude and longitude of all POIs, treating this central point as the anchor for local spatial reasoning. We reimagine the spatial landscape of each cluster by transforming every POI’s global position into polar coordinates. This transformation prioritizes local proximity over global precision, preserving fine-grained spatial relationships that are crucial for feasible recommendations.

To operationalize local geographic coordinates in the quantized token space and address \textbf{Q2}, we design a \textit{\textbf{geo-rotary position encoding}} (GEO-RoPE) inspired by the RoPE mechanism \cite{su2024roformer}. GEO-RoPE encodes each POI’s spatial offset via complex-valued rotations applied directly to the semantic codeword embedding. This design encodes geographic proximity as geometry-aware transformations in the representation space, ensuring that spatially nearby POIs share similar embeddings while preserving semantic integrity through orthogonal rotations. By integrating (i) standard semantic codewords and (ii) a geo-codeword that captures local spatial displacement via geometric rotation, Pro-GEO balances semantic dominance and geographic sensitivity within a unified codebook framework. Our contributions are summarized as follows:

1) We propose a geo-centroid local coordinate system that anchors each semantic cluster to a geographic centroid, explicitly preserving fine-grained intra-cluster spatial relationships.

2) We design Geo-RoPE, a geometry-aware encoding mechanism that embeds relative geographic proximity directly into the token space while maintaining semantic structure.

3) We conduct extensive experiments on a large-scale industrial dataset, demonstrating that Pro-GEO reduces average geographic clustering distance by 45.60\% without codebook collapse and improves Hit@50 by 1.87\%.

\section{Related works}
\begin{figure*}[h]
    \centering
    \includegraphics[width=0.95\linewidth]{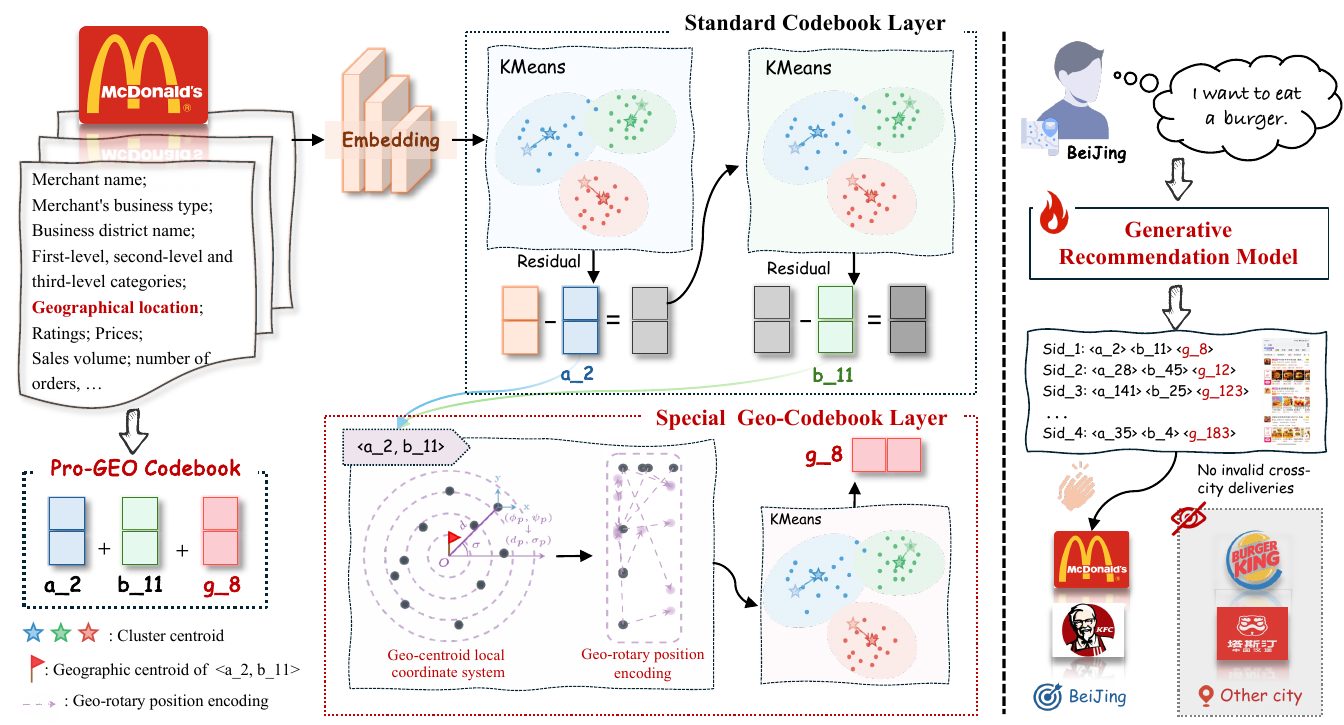}
    \vspace{-0pt}
    \caption{The overview of the ProGEO. It includes two standard codebook layers and a geo-codebook layer. The standard codebook layers use RQ-Kmeans to cluster mixed feature vectors, while the geo-codebook layer enhances spatial features from the residuals of semantic IDs. Each POI’s global coordinates $(\phi_p, \psi_p)$ are converted into geo-center local coordinates $(d_p,\sigma_p)$, with Geo-RoPE applied for geographic modulation. Based on user intent, a fine-tuned LLM generates recommendation candidates, ensuring the exclusion of invalid cross-city deliveries.}
    \vspace{0pt}
    \label{fig:2}
\end{figure*}
\subsection{Item Tokenization}
The tokenization has become the predominant approach in generative recommendation, aiming to compress rich item semantics into fixed-length discrete symbol sequences \cite{li2025survey}. 

The transition from VQ-VAE \cite{van2017neural} to its variant RQ-VAE marked a decisive milestone, enabling SID tokenization to overcome the expressiveness limitation of single-codebook quantization and become integral to modern generative recommender systems. Representative frameworks such as TIGER \cite{rajput2023recommender} and LC-Rec \cite{zheng2024adapting} typically employed hierarchical and residual quantization, followed by sequence modeling over semantic IDs. OneRec \cite{deng2025onerec} adopted a deterministic RQ-Kmeans clustering framework, whose parameter-free and fixed-step nature yields more stable codebook allocation. Zhou et al. \cite{zhou2025hymirec} utilized cosine similarity and directional projection for residual quantization to better preserve semantic angles. PSRQ \cite{wang2025progressive} enhanced semantic preservation by concatenating residual vectors with the original content features at each quantization layer. OneSearch \cite{chen2025onesearch} targeted the core requirement of strong query-item relevance by employing hierarchical quantization techniques like RQ-OPQ to reinforce the preservation of key item attributes. As SID tokenization is extended to more complex and open-world scenarios, adaptation to domain-specific signals becomes critical, such as local life services. Recent works such as OneLoc \cite{wei2025oneloc} and GNPR-SID \cite{wang2025generative} directly embed geographic coordinate information directly into item representations, thereby generating intrinsically location-aware semantic IDs. LGSID \cite{jiang2025llm} introduced an LLM-aligned geographic item tokenization framework that leverages reinforcement learning-based alignment and hierarchical geo-aware quantization to endow semantic IDs with real-world spatial awareness.


\subsection{Generative recommendation}

 Initial research \cite{rajput2023recommender, geng2022recommendation, cui2022m6} explores the direct application of powerful conversational large language models (LLMs) to recommendation systems. Subsequent research focuses on the fine-grained modeling of user interests and multi-objective optimization. GPR \cite{zhang2025gpr} proposed a heterogeneous hierarchical decoder and multi-stage joint training strategies to align user intent modeling with advertising generation. OneRec \cite{deng2025onerec} leveraged session-wise generation and Mixture-of-Experts architectures to produce coherent and diverse video recommendations, while OneRec-V2 \cite{zhou2025onerec} streamlined computational efficiency with a lazy decoder-only architecture and improved preference alignment via direct user feedback. OneRec-Think \cite{liu2025onerec} extends generative recommendation by explicitly integrating dialogue and reasoning capabilities, enabling emotionally adaptive recommendations through reasoning-enhanced reward functions. UniSearch \cite{chen2025unisearch} reimagined search systems with an end-to-end generative framework that unifies semantic identifier generation and multimodal content encoding.

\section{Methodology}

Our Pro-GEO is based on the consensus that geographical proximity critically influences user experience, The framework is shown in Fig.\ref{fig:2}. Let $P = \{p_{1}, p_{2}, ..., p_{N}\}$ represent the set of $N$ POIs, respectively. Each POI is characterized by a tuple $p_N = (s, \phi, \psi)$, where $s$ contains rich semantic information such as category, brand, rating, and price, and $(\phi, \psi)$ denote the geographical coordinates. By mapping POI features into unique identifiers $sid$, the generative recommendation task aims to generate personalized POI suggestions by modeling the conditional probability distribution over the POI set.


\subsection{Semantic Codebook Clustering}
Given POI representations $\mathcal{X} = \{\mathbf{x}_i \in \mathbb{R}^M\}_{i=1}^N$, the conventional RQ-Kmeans algorithm relies on Euclidean distance for similarity measurement, which is ill-suited for high-dimensional embedding spaces. Euclidean distance becomes less discriminative due to the curse of dimensionality, and may not align well with the similarity objectives used in industrial recommender systems.

Inspired by HymiRec \cite{zhou2025hymirec}, we adopt cosine similarity as the clustering metric. At layer $l$, for each residual embedding $\mathbf{r}_i^{\langle l-1 \rangle }$, the cluster assignment is determined: 
\begin{equation}
    j_i^{\langle l \rangle } = \arg\max_{j} \frac{\mathbf{r}_i^{\langle l-1 \rangle }{}^\top \mathbf{c}_j^{\langle l \rangle }}{\|\mathbf{r}_i^{\langle l-1 \rangle }\|_2 \|\mathbf{c}_j^{\langle l \rangle }\|_2},
\end{equation}
where $\mathbf{r}_i^{\langle 0 \rangle } = \mathbf{x}_i$. Cosine similarity focuses on the angular relationship between vectors, offering greater robustness to scale variations. 

Furthermore, we compute the projection residuals to further eliminate information redundancy. For each data point $i$ and layer $l$, we project the residual embedding $\mathbf{r}_i^{\langle l-1 \rangle }$ onto the direction of the assigned centroid $\mathbf{c}_{j_i^{\langle l \rangle }}^{\langle l \rangle }$ and remove this component:
\begin{equation}
    \mathbf{r}_i^{\langle l \rangle } = \mathbf{r}_i^{\langle l-1 \rangle } - 
    \frac{\mathbf{r}_i^{\langle l-1 \rangle \top} \mathbf{c}_{j_i^{\langle l \rangle }}^{\langle l \rangle }}{\|\mathbf{c}_{j_i^{\langle l \rangle }}^{\langle l \rangle }\|_2^2} \mathbf{c}_{j_i^{\langle l \rangle }}^{\langle l \rangle }.
\end{equation}

This operation projects the residual onto the subspace orthogonal to the selected centroid direction, thereby filtering out redundant information while retaining the magnitude of the embedding along the cluster direction. By iteratively applying this process over two layers, we obtain semantic identifiers $(j_i^{\langle 1 \rangle }, j_i^{\langle 2 \rangle })$, which form a hierarchical semantic codebook.

\subsection{Geo-Centroid Local Coordinate Mapping}
Global geographic coordinates present two challenges in local spatial modeling: (1) Minor coordinate variations may correspond to vast distances in the real world; (2) Lack of reference standards supporting relative positioning within regions. To overcome these limitations, we propose a geo-centroid local coordinate system, as shown in Fig.\ref{fig:3}.

\begin{figure}[h]
    \centering
    \includegraphics[width=1\linewidth]{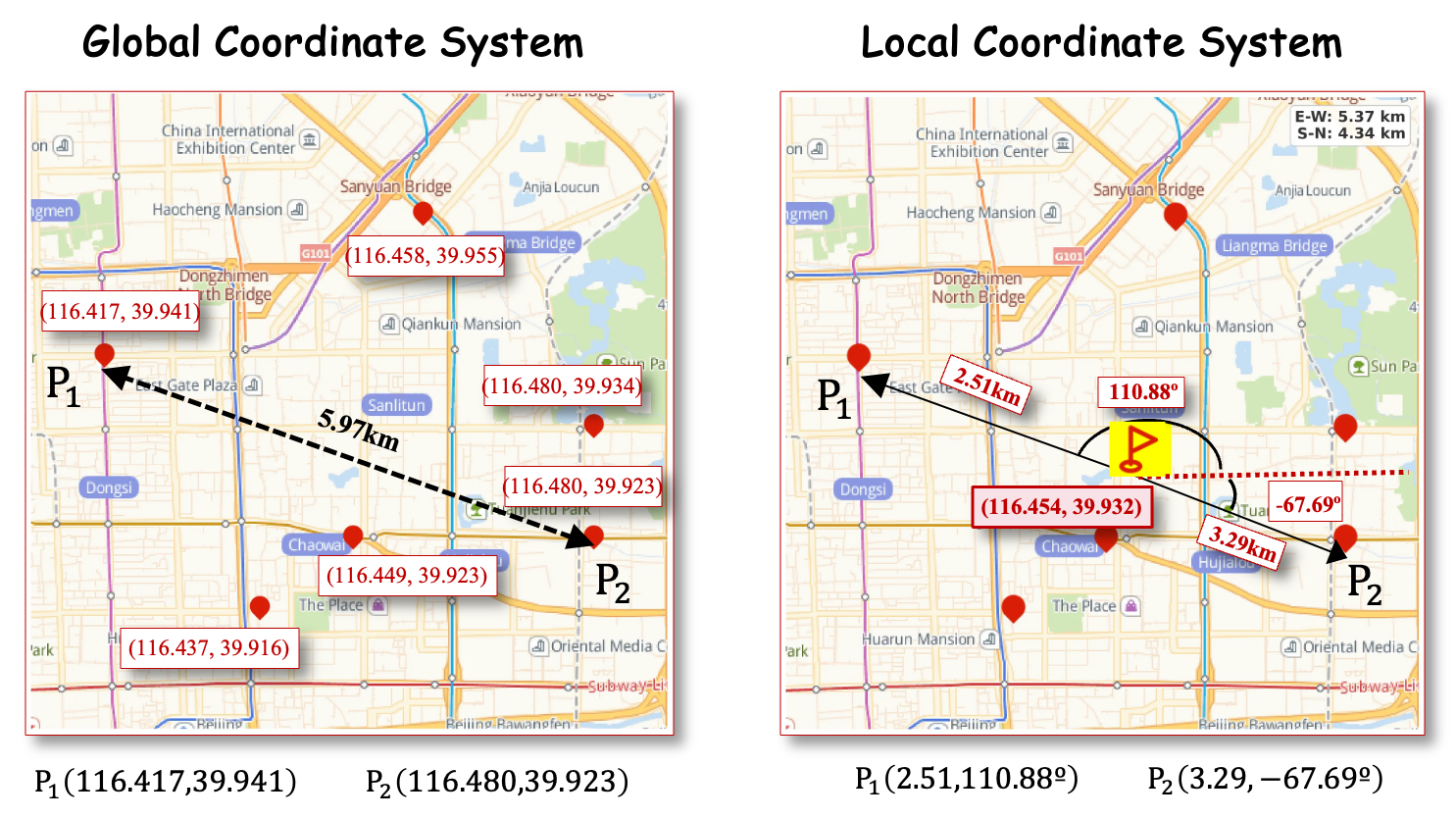}
    \vspace{0pt}
    \caption{Comparison between global and local Coordinate Systems. The global system expresses locations in absolute coordinates, which may obscure local spatial relationships. In contrast, the Geo-Centroid local system anchors positions relative to a centroid, enabling more accurate and interpretable modeling of local spatial patterns.}
    \vspace{0pt}
    \label{fig:3}
\end{figure}

Let $\mathcal{S}_{(j^{\langle 1 \rangle },j^{\langle 2 \rangle })}$ be the set of POIs assigned to the two-layer SID, Each POI $p$ is associated with global geographic coordinates $(\phi_p, \psi_p)$, which represent latitude and longitude. We adopt the geographic centroid $O_{p}^{j}$ of $\mathcal{S}_{(j^{\langle 1 \rangle },j^{\langle 2 \rangle })}$ as the origin of its local coordinate system.

\begin{equation}
O_{p}^{j} = (\phi_0, \psi_0) = \left( \frac{1}{|\mathcal{S}|} \sum_{s \in \mathcal{S}} \phi_s, \frac{1}{|\mathcal{S}|} \sum_{s \in \mathcal{S}} \psi_s \right).
\end{equation}

Each POI within the cluster is then represented by its relative polar coordinates with respect to the centroid:

\begin{equation}
\begin{aligned}
d_p &= 2R \arcsin\left(
    \sqrt{
        \sin^2\left(\frac{\phi_p - \phi_0}{2}\right)
        + \cos\phi_0 \cos\phi_p \sin^2\left(\frac{\psi_p - \psi_0}{2}\right)
    }
\right), \\
\sigma_p &= \arg\left(x + i y\right),
\end{aligned}
\label{eq:polar}
\end{equation}
where $R$ is the mean radius of the Earth, typically set to $6371$ kilometers, $\arg(\cdot)$ denotes the principal value of the argument, and 
\begin{equation}
    \begin{aligned}
x &= \cos\phi_0 \cdot \sin\phi_p - \sin\phi_0 \cdot \cos\phi_p \cdot \cos(\psi_p - \psi_0), \\
y &= \sin(\psi_p - \psi_0) \cdot \cos\phi_p.
\end{aligned}
\end{equation}

This representation provides two key advantages: (1) the radial distance $d_p$ encodes accessibility relative to the cluster center, and (2) the azimuth angle $\sigma_p$ captures directional distribution within the localized region.

\subsection{Geo-Rotary Position Encoding}

RoPE \cite{su2024roformer} is a self-attention position encoding mechanism originally designed for Transformers, which models the relative position between sequence elements via orthogonal rotations. We extend RoPE to the spatial domain, formalizing geographic proximity as rotary operations in high-dimensional feature space. 

Given a POI $p$, we denote its second-layer residual vector as $\mathbf{r}_p^{\langle 2 \rangle} \in \mathbb{R}^{M}$, and its geo-centroid local coordinates as $(d_p, \sigma_p)$. Then, define the block-diagonal rotation matrix as

\begin{equation}
    \mathcal{R}(\theta_p) = \text{diag}\underbrace{(R(\theta_p), R(\theta_p), \ldots, R(\theta_p))}_{m \text{ times}} \in \mathbb{R}^{M \times M},
\end{equation}
where $m=M/2$, and
\begin{equation}
    R(\theta_p) = \begin{bmatrix} \cos\theta_p & -\sin\theta_p \\ \sin\theta_p & \cos\theta_p \end{bmatrix}.
\end{equation}

Unlike the original RoPE, the frequency parameter $\theta_p$ determines the wavelength distribution and controls the receptive field over sequential positions, our codebook does not model contextual relationships between tokens within a residual vector. Instead, it focuses on modeling relative spatial relationships between different POI residual vectors, $\theta_p \in \{d_p, \sigma_p\}$. To further reinforce spatial sensitivity, we propose the forward and reverse rotation strategy.

Taking $\theta_p = \sigma_p$ as an example, the Geo-RoPE transformation $\mathcal{T}_{\theta_{p}}: \mathbb{R}^{M} \rightarrow \mathbb{R}^{2M}$ is defined as
\begin{equation}
   \mathbf{r}_{p}^{\langle 2 \rangle\prime} =  \mathcal{T}_{\sigma_p}(\mathbf{r}_p^{\langle 2 \rangle}) = \begin{bmatrix} \mathcal{R}(\sigma_p) \mathbf{r}_p^{\langle 2 \rangle} \\ \mathcal{R}(-\sigma_p) \mathbf{r}_p^{\langle 2 \rangle} \end{bmatrix} \in \mathbb{R}^{2M}.
\end{equation}

For any two vectors $\mathbf{r}_{i}^{\langle 2 \rangle}$ and $\mathbf{r}_{j}^{\langle 2 \rangle}$, the cosine difference can be defined:
\begin{equation}
    D_{\cos}\left(\mathbf{r}_{i}^{\langle 2 \rangle},\, \mathbf{r}_{j}^{\langle 2 \rangle}\right) = 1- \frac{\langle \mathbf{r}_{i}^{\langle 2 \rangle},\, \mathbf{r}_{j}^{\langle 2 \rangle} \rangle}
{\left\|\mathbf{r}_{i}^{\langle 2 \rangle}\right\|\, \left\|\mathbf{r}_{j}^{\langle 2 \rangle}\right\|} ,
\end{equation} 
and the change of before and after the geo-RoPE transformation is given by:
\begin{align}
\Delta D_{\cos} 
&= D_{\cos}\left(\mathbf{r}_{i}^{\langle 2 \rangle\prime},\, \mathbf{r}_{j}^{\langle 2 \rangle\prime}\right) 
   - D_{\cos}\left(\mathbf{r}_{i}^{\langle 2 \rangle},\, \mathbf{r}_{j}^{\langle 2 \rangle}\right) \nonumber \\
&= 2 \cdot \frac{\langle \mathbf{r}_{i}^{\langle 2 \rangle},\, \mathbf{r}_{j}^{\langle 2 \rangle} \rangle}
{\left\|\mathbf{r}_{i}^{\langle 2 \rangle}\right\|\, \left\|\mathbf{r}_{j}^{\langle 2 \rangle}\right\|} \cdot 
\sin^2\left( \frac{\Delta \sigma}{2} \right)
\label{eq:11},
\end{align}
where $\Delta \sigma = \sigma_i - \sigma_j$. Restricting $\sigma$ in $[-\pi/2, +\pi/2]$ ensures that $\Delta \sigma$ captures the precise rotational difference between POIs. $\Delta D_{\cos}$ is minimized when $\Delta\sigma = 0$ and maximized when $\Delta\sigma \in [-\pi, +\pi]$. This property ensures that the transformation naturally preserves similarity for spatially proximate POIs while appropriately separating distant ones. The complete proof of Eq.~\ref{eq:11} is provided in Appendix~A.

Importantly, Eq.~\ref{eq:11} decomposes the change into two interpretable factors:

\begin{equation}
\Delta D_{\cos} \propto 
    \underbrace{\sin^2\left(\frac{\Delta\sigma}{2}\right)}_{\begin{array}{c}\substack{\scriptsize\text{Positional}\\\scriptsize\text{difference factor}}\end{array}}
    \times 
    \underbrace{\frac{\langle \mathbf{r}_{i}^{\langle 2 \rangle},\, \mathbf{r}_{j}^{\langle 2 \rangle} \rangle}
{\left\|\mathbf{r}_{i}^{\langle 2 \rangle}\right\|\, \left\|\mathbf{r}_{j}^{\langle 2 \rangle}\right\|}}_{\begin{array}{c}\substack{\scriptsize\text{Semantic}\\\scriptsize\text{similarity factor}}\end{array}}
\end{equation}

This formulation demonstrates the cosine difference change of $\sigma$ is not determined by spatial separation alone, but by the \textit{product} of spatial and semantic factors. For semantically similar POIs (high cosine similarity), small $\Delta\sigma$ preserves their proximity while large $\Delta\sigma$ appropriately separates them. Conversely, for semantically dissimilar POIs, the transformation has minimal effect regardless of spatial configuration, preventing geographic signals from corrupting semantic distinctions. Analogously, this equally applies to $d$ normalized to $[0,\pi]$.

Finally, the rotated second-layer residual vector $\overrightarrow{\mathbf{r}_p^{\langle 2 \rangle}}$ can be obtained:
\begin{equation}
\overrightarrow{\mathbf{r}_p^{\langle 2 \rangle}} =
\begin{bmatrix}
    \mathcal{T}_{\alpha \cdot\sigma_p}(\mathbf{r}_p^{\langle 2 \rangle})\\
     \mathcal{T}_{\beta \cdot d_p}(\mathbf{r}_p^{\langle 2 \rangle}) 
\end{bmatrix} \in \mathbb{R}^{4M}
\end{equation}
where $\alpha, \beta \in \mathbb{R}^+$ are hyperparameters controlling the contributions of the azimuth angle $\sigma_p$ and radial distance $d_p$, respectively.

The transformed embeddings $\overrightarrow{\mathbf{r}_p^{\langle 2 \rangle}}$ are subsequently subjected to $K$-means clustering to yield the third-level identifier. The complete semantic-geographic identifier is then constructed as:
\begin{equation}
\text{SID}_{p} = (j^{\langle 1 \rangle}, j^{\langle 2 \rangle}, j^{\langle 3 \rangle}),
\end{equation}

The first two layers encode semantic structure and the third layer integrates relative geographic information via Geo-RoPE-enhanced clustering. This hierarchical SID framework ensures that semantic and spatial signals are jointly represented, supporting generative models with discrete tokens that reflect both behavioral and geographic context.

\section{Experiments}
In this section, extensive experiments are conducted on an industrial dataset to evaluate the effectiveness of our proposed Pro-GEO. Specifically, we aim to address the following questions:

\textit{\textbf{Q1:} How does Pro-GEO perform compared to existing state-of-the-art methods in the industrial dataset?}

\textit{\textbf{Q2:} What is the contribution of each component in Pro-GEO?}

\textit{\textbf{Q3:} How do different hyper-parameter settings affect the performance of Pro-GEO?}

\textit{\textbf{Q4:} How do geographical clustering in the codebook and visualization of recommendation results provide evidence for the effectiveness of Pro-GEO?}
\subsection{Experimental Setup}
The dataset, implementation details, evaluation metrics, and comparison methods are detailed in Appendix B.

\begin{table*}[h]
\centering
\caption{Performance comparison of Pro-GEO and state-of-the-art methods on quantization and recommendation metrics.
The best values are highlighted in \textbf{bold}, while the second-best values are denoted by \underline{\textit{italic underline}}. 
Positive improvements are shown in \colorbox{green}{\textcolor{black}{positive}} and negative improvements in \colorbox{pink}{\textcolor{black}{negative}}. The unit for Avg. Dist., p90 Dist., and p95 Dist. is kilometers (km). $^\dagger$ indicates that the RQ-OPQ codebook size is set to [512, 512, 512, 256, 256]. For fair comparison, the clustering distance metric is all computed based on the third-layer codebook. *** 
denotes statistically significant improvement over the second-best model RQ-Kmeans ($p < 0.001$). $\uparrow$ and $\downarrow$ indicate that higher and lower values are better for the corresponding evaluation metrics, respectively.}
\label{table:1}
\small
\begin{tabular}{l|ccccc|ccccc}
\toprule
\multirow{2}{*}{Method}  &  \multicolumn{5}{c}{Quantization metrics} & \multicolumn{4}{c}{Recommendation metrics} \\ \cmidrule(lr){2-6} \cmidrule(lr){7-11}
 & CUR($\uparrow$) & ICR($\uparrow$) & Avg. Dist.($\downarrow$) & p90 Dist.($\downarrow$) & p95 Dist.($\downarrow$) & Hit@1($\uparrow$) & Hit@50($\uparrow$) & Hit@100($\uparrow$) & NDCG@50($\uparrow$) & NDCG100($\uparrow$) \\
\midrule
RQ-VAE \cite{rajput2023recommender}   & \textbf{8.02}\%   & 55.87\%   & 64.91  & 109.88 & 243.89 &  0.0887 & 0.3677 & 0.4006 & 0.1800 & 0.1855 \\
RQ-Kmeans  \cite{deng2025onerec}  & 7.26\%   & 55.07\%   & 47.73  & 79.94  & 167.87 & \underline{\textit{0.0978}} & \underline{\textit{0.4056}} & \underline{\textit{0.4401}} & \underline{\textit{0.1989}} & \underline{\textit{0.2047}} \\
Zhou et al. \cite{zhou2025hymirec}  & \textbf{8.02}\%   & \underline{\textit{57.65\%}}   & 46.76  & \underline{\textit{75.26}}  & \underline{\textit{154.93}} & 0.0943 & 0.4031 & 0.4384 & 0.1950 & 0.2009 \\
RQ-OPQ $^\dagger$ \cite{chen2025onesearch}  & 0.00\%   & \textbf{97.70}\%   & \underline{\textit{46.71}}  & 77.46  & 161.10 &   0.0625 & 0.1085 &  0.1085  & 0.0859 & 0.0859 \\
Pro-GEO   & \textbf{7.67\%} & 57.08\% & \textbf{25.41} & \textbf{44.91} & \textbf{75.57} & \textbf{0.0998} & \textbf{0.4243$^{***}$} & \textbf{0.4572$^{***}$} & \textbf{0.2073$^{***}$} & \textbf{0.2128$^{***}$}\\
\midrule
\textit{Imp.} 
& \colorbox{pink}{\textcolor{black}{-0.35\%}} 
& \colorbox{pink}{\textcolor{black}{-40.62\%}} 
& \colorbox{green}{\textcolor{black}{-21.30}} 
& \colorbox{green}{\textcolor{black}{-30.35}} 
& \colorbox{green}{\textcolor{black}{-79.36}} 
& \colorbox{green}{\textcolor{black}{+0.20\%}} 
& \colorbox{green}{\textcolor{black}{+1.87\%}} 
& \colorbox{green}{\textcolor{black}{+1.71\%}} 
& \colorbox{green}{\textcolor{black}{+0.84\%}} 
& \colorbox{green}{\textcolor{black}{+0.81\%}}\\ 
\bottomrule
\end{tabular}
\end{table*} 

\subsection{Performance Analysis (Q1)}
Table \ref{table:1} presents a comprehensive comparison of Pro-GEO and several state-of-the-art methods across both quantization and recommendation metrics. In terms of quantization, Pro-GEO achieves a CUR of 7.67\% and an ICR of 57.08\%, comparable to the best-performing methods, while significantly outperforming others in cluster compactness. Specifically, Pro-GEO yields the lowest average cluster diameter, as well as the lowest p90 Dist. and p95 Dist., indicating highly compact and well-separated clusters. This demonstrates the model's ability to effectively leverage geographic features for fine-grained spatial grouping. On recommendation metrics, Pro-GEO consistently surpasses all methods, achieving the highest scores on Hit@1 (0.0998), Hit@50 (0.4243), Hit@100 (0.4572), NDCG@50 (0.2073), and NDCG@100 (0.2128). These improvements are particularly notable compared to RQ-Kmeans and Zhou et al., which represent strong methods in both clustering and recommendation domains. Overall, these results validate the superiority of Pro-GEO in integrating geographic information for local service recommendation, achieving significant improvements in clustering based on joint semantic and geographic features, as well as top-N recommendation performance.

\subsection{Ablation Study (Q2)}
To evaluate the contribution of each component within the Pro-GEO framework, we conduct comprehensive ablation studies, guided by the following sub-questions:

\textit{\textbf{Q2.1:} Which type of geographic information serves as the most effective reference standard for local service recommendation?}

\textit{\textbf{Q2.2:} which stage of the model architecture does the integration of geographic information yield the greatest benefit?}

\textit{\textbf{Q2.3:} How can geographic information be supplemented to maximize overall cluster effectiveness?}

\textit{\textbf{Q2.4:} Which specific types of local geographic feature provide the most significant improvement?}

\textit{\textbf{Q2.5:} Which strategy best resolves codebook conflicts in local service recommendation systems?}

\subsubsection{Effectiveness of Geographic Types} To investigate the impact of global and local geographic information, we design the following experimental settings at the third-level clustering layer: (i) \textbf{w/ global}: global geographic coordinates (latitude-longitude) as input embedding; (ii) \textbf{w/ local}: geo-centroid local geographic coordinates ($d, \sigma$) as input embedding; (iii) \textbf{w/ R-global}: applying Geo-RoPE to enhance the representation of global geographic coordinates; (iv) \textbf{w/ R-local}: applying Geo-RoPE to enhance the representation of local geographic coordinates.

The experimental results in Fig.\ref{fig:12} reveal several important patterns regarding geographic representation for clustering. First, local geographic embeddings consistently outperform global counterparts across all metrics, confirming the advantage of modeling fine-grained spatial relationships. Specifically, the CUR and ICR of local methods are markedly higher, while cluster diameters are substantially reduced, indicating more efficient and compact clustering. Second, the application of Geo-RoPE representation learning brings significant gains regardless of coordinate type, but its effect is most pronounced in the local setting. Our Geo-RoPE-based local method achieves the highest CUR (7.67\%) and ICR (57.08\%), and the compact clusters, outperforming both the vanilla local and Geo-RoPE-based global approaches. 

\subsubsection{Impact of Geographic Feature Integration Stage} We compare multiple integration strategies to determine the most effective stage for injecting Geo-RoPE-based geographic signals. (i) \textbf{w/ only second layer}: Geo-RoPE encoding is applied at the second clustering layer; (ii) \textbf{w/ only third layer}: Geo-RoPE encoding is applied at the third clustering layer; (iii) \textbf{w/ both layers}: Geo-RoPE encoding is applied at both the second and third clustering layers.

As shown in Fig. \ref{fig:13}, integrating Geo-RoPE only at the third clustering layer (our setting) yields the highest CUR (7.67\%) and ICR (57.08\%), indicating superior clustering efficiency. This configuration also achieves notably compact geographical clustering (Avg. Dist.: 25.41), outperforming single-layer and dual-layer integration in terms of spatial grouping. Applying Geo-RoPE at both layers further reduces the geographical distances within clusters, suggesting enhanced spatial compactness. However, this comes at the expense of lower CUR (6.36\%) and ICR (52.95\%), reflecting a trade-off between geographic compactness and overall clustering effectiveness. In contrast, integrating Geo-RoPE only at the second layer results in the least favorable outcomes across all metrics, underscoring the limited impact of early-stage geographic feature injection.
\begin{figure}
    \centering
    \includegraphics[width=0.95\linewidth]{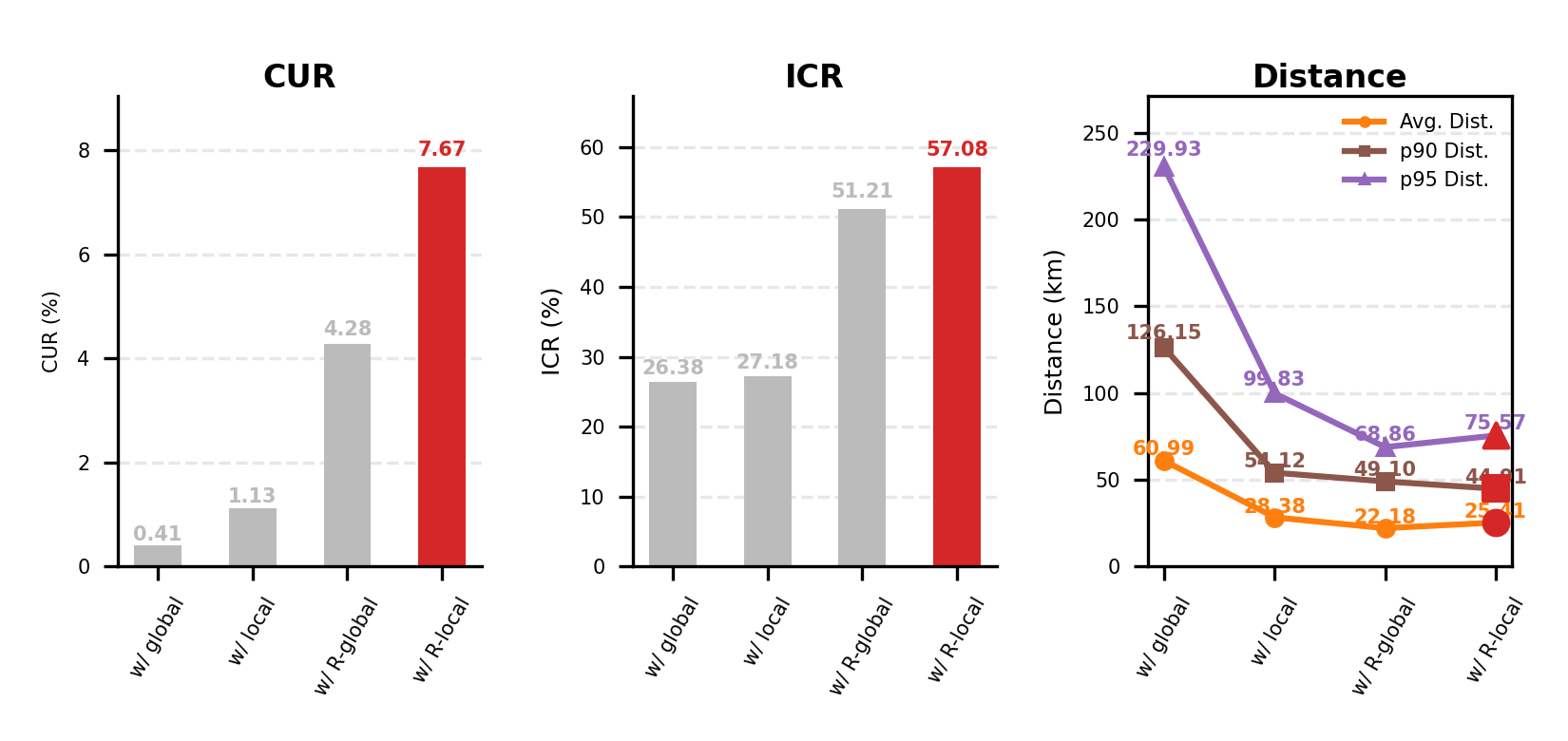}
    \vspace{-10pt}
    \caption{Comparison of global and local geographic representation strategies.}
    \vspace{-0pt}
    \label{fig:12}
\end{figure}
\begin{figure}
    \centering
    \includegraphics[width=0.95\linewidth]{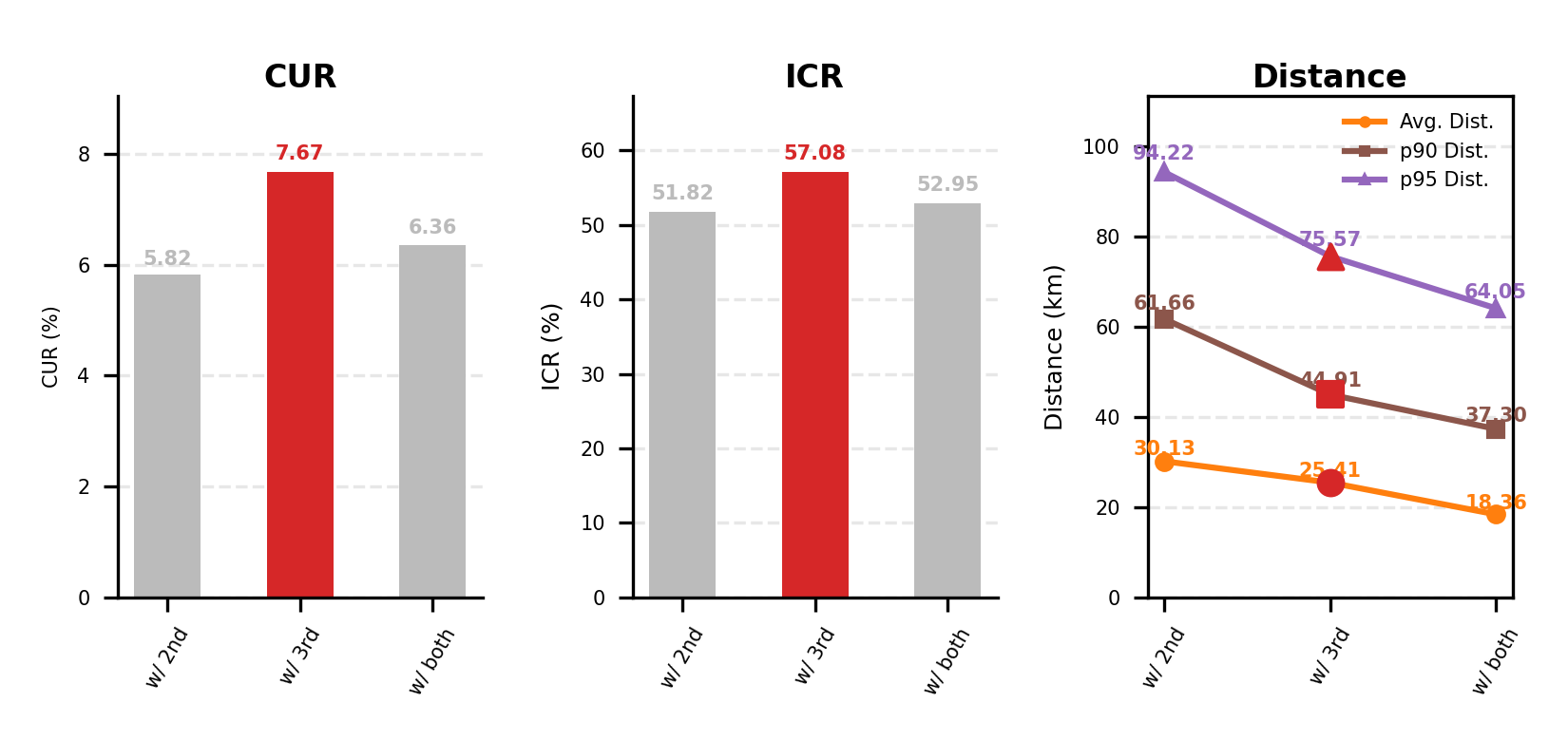}
    \vspace{-10pt}
    \caption{Comparison of Geo-RoPE integration strategies at different clustering layers. }
    \vspace{-0pt}
    \label{fig:13}
\end{figure}

\subsubsection{Enhancement Strategies for Geographic Information} To validate the superiority of Geo-RoPE, we compare several geographic information enhancement methods: (i) \textbf{w/ concat}: geographic information is appended as supplementary features, expanding the vector dimension from 128 to 130; (ii) \textbf{w/ add}: geographic information is directly added to the original 128-dimensional vector, keeping the dimension unchanged; (iii) \textbf{w/ Geo-RoPE}: Geographic information is enhanced using the Geo-RoPE encoding mechanism.
\begin{table*}[htbp]
\centering
\caption{
Comparison of different geographic attribute combinations on quantization and recommendation metrics. 
The best values are highlighted in \textbf{bold}, while the second-best values are denoted by \underline{\textit{italic underline}}. The unit for Avg. Dist., p90 Dist., and p95 Dist. is kilometers (km). $^\dagger$ indicates our proposed Pro-GEO method. $\uparrow$ and $\downarrow$ indicate that higher and lower values are better for the corresponding evaluation metrics, respectively.}
\footnotesize
\label{tab:comparison}
\begin{tabular}{l|ccccc|ccccc}
\toprule
\multirow{2}{*}{Config}  &  \multicolumn{5}{c}{Quantization metrics} & \multicolumn{5}{c}{Recommendation metrics} \\ \cmidrule(lr){2-6} \cmidrule(lr){7-11}
 & CUR($\uparrow$) & ICR($\uparrow$) & Avg. Dist.($\downarrow$) & p90 Dist.($\downarrow$) & p95 Dist.($\downarrow$) & Hit@1($\uparrow$) & Hit@50($\uparrow$) & Hit@100($\uparrow$) & NDCG@50($\uparrow$) & NDCG100($\uparrow$) \\
\midrule
w/ $\mathcal{R}(d^{+})$ & \underline{\textit{8.09\%}} & 57.90\% & 40.52 & 66.39 & 123.34 & 0.0950 & 0.3934 & 0.4257 & 0.1931 & 0.1985 \\
w/ $\mathcal{R}(d^{-})$ & \textbf{8.11\%} & \textbf{58.07\%} & 40.20 & 66.57 & 124.50 & 0.0966 & 0.3958 & 0.4277 & 0.1950 & 0.2003 \\
w/ $\mathcal{R}(d^{+},d^{-})$ & 7.94\% & 57.67\% & 38.23 & 65.21 & 120.45 & 0.0974 & 0.4057 & 0.4374 & 0.1989 & 0.2042 \\
w/ $\mathcal{R}(\sigma^{+})$ & \textbf{8.42\%} & \underline{\textit{58.50\%}} & \underline{\textit{26.42}} & \underline{\textit{43.82}} & \underline{\textit{75.22}} & 0.0865 & 0.3824 & 0.4082 & 0.1850 & 0.1893 \\
w/ $\mathcal{R}(\sigma^{-})$ & 8.35\% & \textbf{58.66\%} & 26.61 & 45.18 & 76.83 & 0.0896 & 0.3856 & 0.4115 & 0.1879 & 0.1923 \\
w/ $\mathcal{R}(\sigma^{+},\sigma^{-})$ & 7.53\% & 56.79\% & \textbf{20.14} & \textbf{37.04} & \textbf{62.96} & 0.0949 & \underline{\textit{0.4206}} & \underline{\textit{0.4494}} & \underline{\textit{0.2025}} & \underline{\textit{0.2073}} \\
w/ $\mathcal{R}(\sigma^{+},d^{+})$ & 8.02\% & 57.69\% & 27.46 & 47.37 & 81.62 & \underline{\textit{0.0982}} & 0.4146 & 0.4447 & 0.2031 & 0.2082 \\
w/ $\mathcal{R}(\sigma^{+},\sigma^{-},d^{+},d^{-})$ $^\dagger$ & 7.67\% & 57.08\% & 25.41 & 44.91 & 75.57 & \textbf{0.0998} & \textbf{0.4243} & \textbf{0.4572} & \textbf{0.2073} & \textbf{0.2128} \\
\bottomrule
\end{tabular}
\end{table*}

The comparison of geographic information enhancement strategies in Fig. \ref{fig:14} reveals clear differences in clustering performance and spatial representation. The Geo-RoPE approach demonstrates a distinct advantage over conventional methods, as evidenced by its highest CUR and ICR. This indicates that Geo-RoPE not only enables the clustering model to identify more meaningful spatial groupings, but also improves the overall coverage and efficiency of the clustering process. In contrast, the concatenation strategy, while yielding the smallest average cluster diameter (Avg. Dist.: 13.46), shows moderate clustering utilization. This suggests that simply appending geographic features increases local compactness but fails to fully exploit the spatial relationships, potentially leading to over-fragmentation and reduced cluster interpretability. The addition method, which directly merges geographic information into the existing feature space, performs worst across all metrics, highlighting its inability to effectively represent geographic signals and its tendency to dilute spatial information. The results highlight that Geo-RoPE is essential for maximizing both clustering effectiveness and spatial coherence.
\begin{figure}
    \centering
    \includegraphics[width=0.95\linewidth]{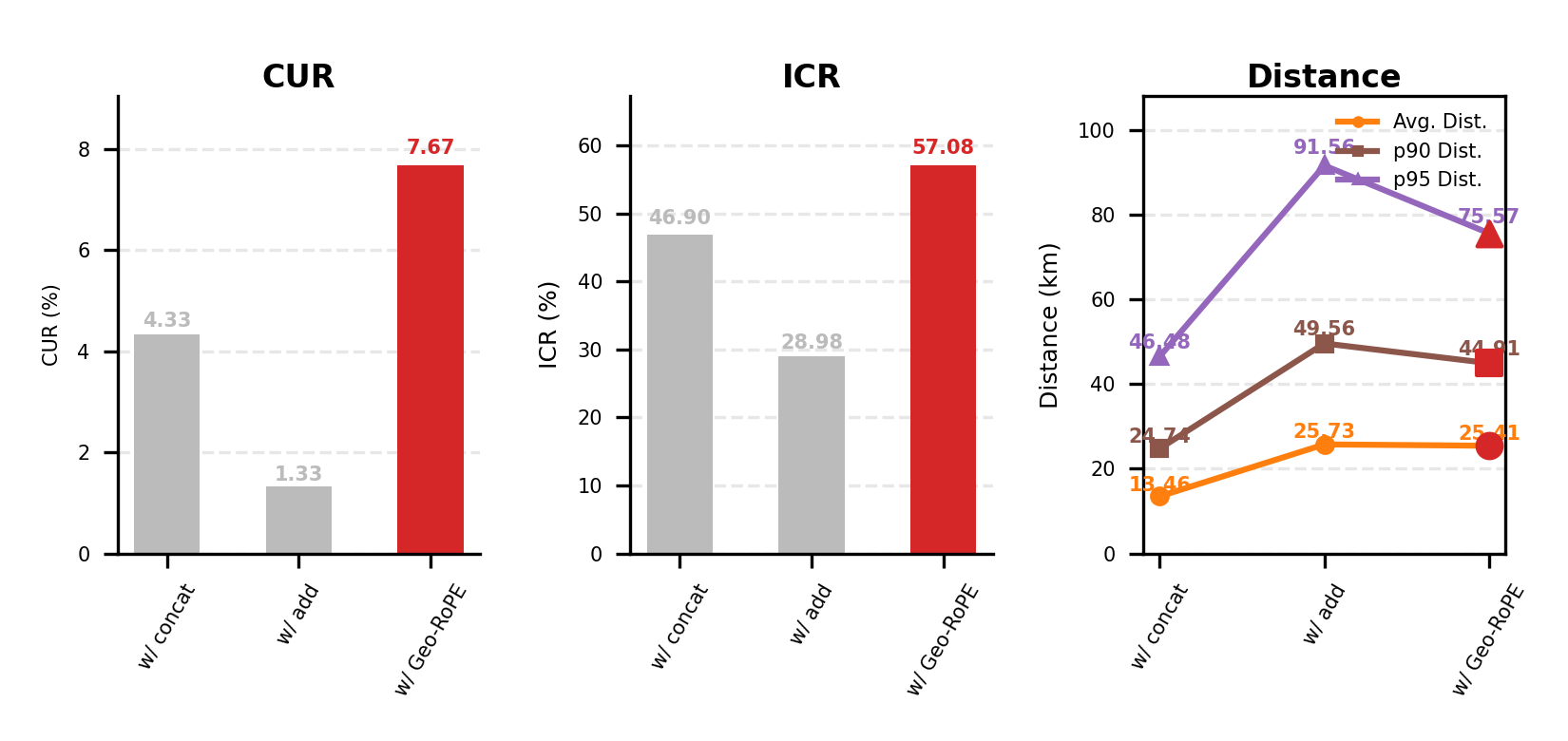}
    \vspace{-10pt}
    \caption{Comparison of geographic information enhancement strategies for clustering.}
    \vspace{-0pt}
    \label{fig:14}
\end{figure}

\begin{table}[t]
    \centering
    \caption{Comparison of different strategies for resolving codebook conflicts in local service recommendation.}
    \begin{tabularx}{0.8\linewidth}{>{\raggedright\arraybackslash}X|c|c}
    \toprule
       Config  & Hit@1 & Imp. \\ \midrule
       Random & 3.23\% & - \\
       Layer-4 hard coding  &  5.97\% & +2.74\%  \\
       2\_256\_opq & 6.38\% & +3.15\% \\ \midrule
       Closest match (Ours) & \textbf{9.98\%} & \textbf{+6.55\%} \\
       \bottomrule
    \end{tabularx}
    \label{tab:placeholder}
\end{table}

\subsubsection{Contribution of Geographic Attributes} We conduct a fine-grained analysis by incrementally adding various local geographic attributes. Specifically, we evaluate the following settings, where $\mathcal{R}$ denotes the Geo-RoPE operation: (i) \textbf{w/ $\mathcal{R}(d^{+})$}; (ii) \textbf{w/ $\mathcal{R}(d^{-})$}; (iii) \textbf{w/ $\mathcal{R}(d^{+},d^{-})$}; (iv) \textbf{w/ $\mathcal{R}(\sigma^{+})$}; (v) \textbf{w/ $\mathcal{R}(\sigma^{-})$}; (vi) \textbf{w/ $\mathcal{R}(\sigma^{+},\sigma^{-})$}; (vii) \textbf{w/ $\mathcal{R}(\sigma^{+},d^{+})$}; (viii) \textbf{w/ $\mathcal{R}(\sigma^{+},\sigma^{-},d^{+},d^{-})$}.

Table~\ref{tab:comparison} presents a comparison of different geographic attribute combinations. The results reveal several key insights: 
First, incorporating only the radial distance ($d^{+}$ or $d^{-}$) leads to moderate improvements in quantization and recommendation metrics, while geographic clustering performance remains suboptimal.
Second, adding azimuth angle ($\sigma^{+}$, $\sigma^{-}$) significantly reduces the cluster diameters, indicating compact spatial groups, but the recommendation scores do not reach their highest values.
Third, the configuration $w/ \mathcal{R}(\sigma^{+},\sigma^{-})$ further enhances spatial compactness, with Avg. Dist reduced to 20.14 and p90 Dist to 37.04. Notably, this forward and reverse rotation encoding ($\sigma^{+},\sigma^{-}$) outperforms single-directional variants in both spatial clustering and recommendation metrics, yielding higher Hit@50, Hit@100, and NDCG scores. This indicates that modeling spatial relationships from both angular directions allows the model to capture more complex and symmetric spatial patterns, thereby mitigating directional bias and improving overall performance. Finally, the configuration $w/ \mathcal{R}(\sigma^{+},\sigma^{-},d^{+},d^{-})$ achieves the best overall results, with consistently high recommendation metrics (e.g., Hit@50 = 0.4243, NDCG@50 = 0.2073) and superior geographic clustering performance. This synergy between radial distance and polar angle attributes in the relative polar coordinate system is critical for both effective spatial quantization and accurate recommendation.


\subsubsection{Resolution Strategies for Codebook Conflicts}

Table~\ref{tab:placeholder} provides a comparative evaluation of different strategies for resolving codebook conflicts in local service recommendation. The results demonstrate clear performance disparities among the tested methods. The Random baseline, which assigns POIs by randomly matching each sid, yields the lowest Hit@1 score (3.23\%), indicating limited recommendation accuracy. Layer-4 hard coding, which assigns a unique codeword to each POI at the fourth layer, improves Hit@1 to 5.97\%, suggesting that explicit codeword separation enhances retrieval precision but still suffers from insufficient semantic alignment. The 2\_256\_opq configuration, which incorporates two additional OPQ layers, further increases Hit@1 to 6.38\%, reflecting the benefit of more expressive quantization in reducing codebook ambiguities. Most notably, the Closest match strategy—our proposed approach—achieves a substantial improvement, reaching a Hit@1 score of 9.78\%. By leveraging geographic proximity for POI assignment, this method effectively resolves codebook conflicts and delivers the highest recommendation accuracy, with a relative improvement of +6.55\% over the Random baseline.

\subsection{Parameter Analysis (Q3)}
We conduct a comprehensive hyperparameter study to investigate the effect of rotation scale parameters in the Geo-RoPE process, specifically focusing on $d$ and $\sigma$ components. The hyperparameter configurations are denoted as $(\alpha, \beta)$, where $\alpha$ and $\beta$ control the scaling factors for $d$ and $\sigma$, respectively. Eight configurations are evaluated: $(0,0)$, $(0.25,0.25)$, $(0.25,0.5)$, $(0.5,0.25)$, $(0.5,0.5)$, $(0.5,1)$, $(1,0.5)$, and $(1,1)$.

Fig.~\ref{fig:best_config} presents the sensitivity analysis of the rotation scale parameters $(\alpha, \beta)$ in the Geo-RoPE process across six evaluation metrics. As $(\alpha, \beta)$ increase, all distance-based metrics  (average distance, p90 distance, and p95 distance) exhibit a pronounced decline, suggesting that larger rotation scales facilitate more compact and discriminative spatial clustering. Specifically, the average distance decreases from 45.72~km at $(0,0)$ to 25.41~km at $(0.5,0.5)$, and the p95 distance drops from 148.27~km to 75.57~km. These trends underscore the effectiveness of Geo-RoPE in capturing fine-grained spatial relationships. Concurrently, recommendation metrics also improve with increasing rotation scales. Hit@1 rises from 0.0908 at $(0,0)$ to 0.0998 at $(0.5,0.5)$, while NDCG@50 and NDCG@100 reach their highest values of 0.2073 and 0.2128, respectively, at the same configuration. Notably, the $(0.5,0.5)$ setting, highlighted in red in all plots, consistently achieves optimal or near-optimal performance across both quantization and recommendation tasks. These results collectively demonstrate that tuning rotation scale parameters in Geo-RoPE is crucial for balancing spatial compactness and recommendation accuracy, and that the configuration $(0.5,0.5)$ offers the best trade-off in our experiments.

\begin{figure}[htbp]
\centering
\includegraphics[width=\linewidth]{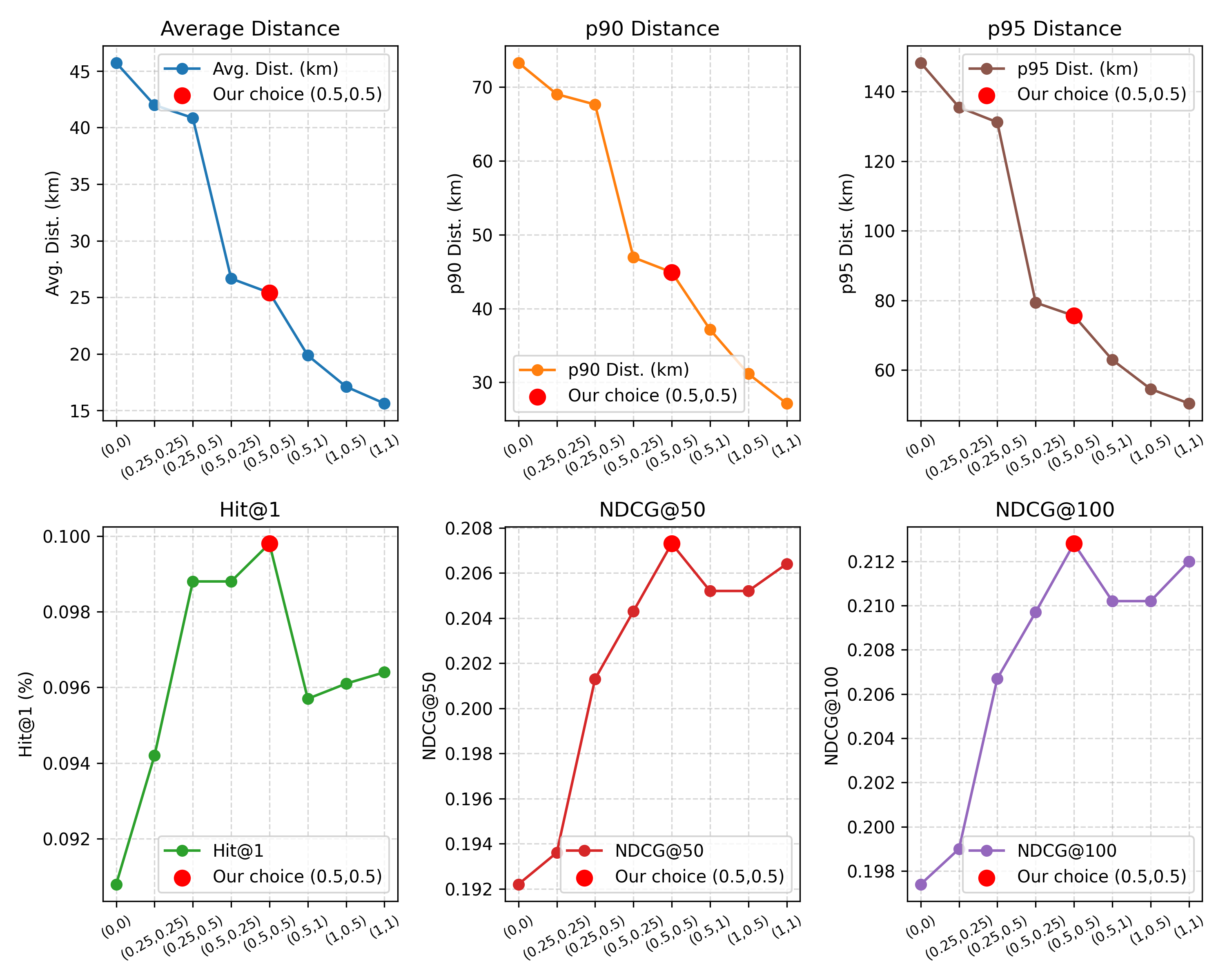}
\vspace{-10pt}
\caption{
Hyperparameter sensitivity analysis.
}
\vspace{-0pt}
\label{fig:best_config}
\end{figure}

\subsection{Case study (Q4)}
\subsubsection{Visualization of codebook geographical clustering}

Fig.~\ref{fig:case_geo_visualization} illustrates the effects of different geographic attributes on the final three-layer codebook clustering. When the radial distance $d$ is used alone, clusters tend to be spatially dispersed, resulting in distinct blocks that may overlook underlying directional relationships among POIs. Conversely, using the azimuth angle $\sigma$ alone causes clusters to aggregate POIs with similar geographic orientation, which may lead to over-grouping along certain directions. Our proposed approach, which jointly incorporates both $d$ and $\sigma$, effectively balances these two aspects. The resulting clusters are both spatially compact and directionally coherent, avoiding the drawbacks of excessive dispersion or over-aggregation. 

\begin{figure}
    \centering
    \includegraphics[width=\linewidth]{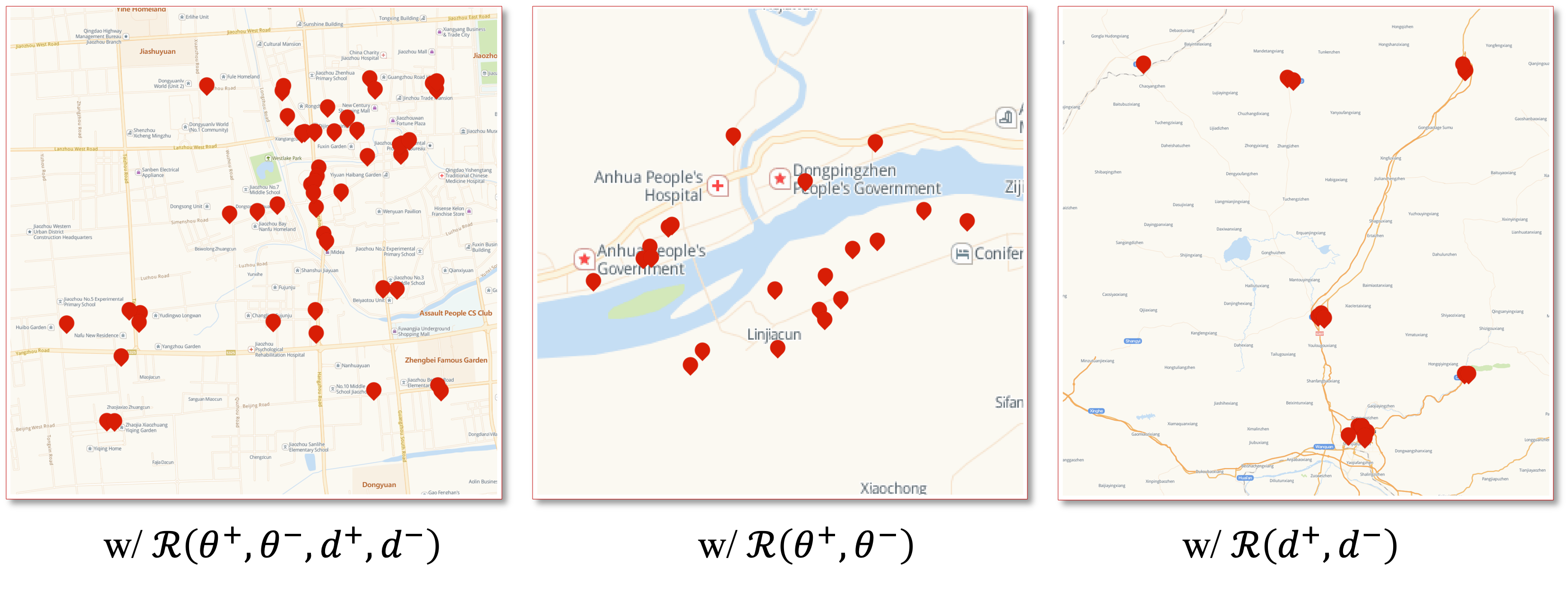}
    \caption{Visualization of codebook geographical clustering with different attribute configurations.}
    \label{fig:case_geo_visualization}
    \vspace{-0pt}
\end{figure}
\begin{figure}
    \centering
    \includegraphics[width=0.95\linewidth]{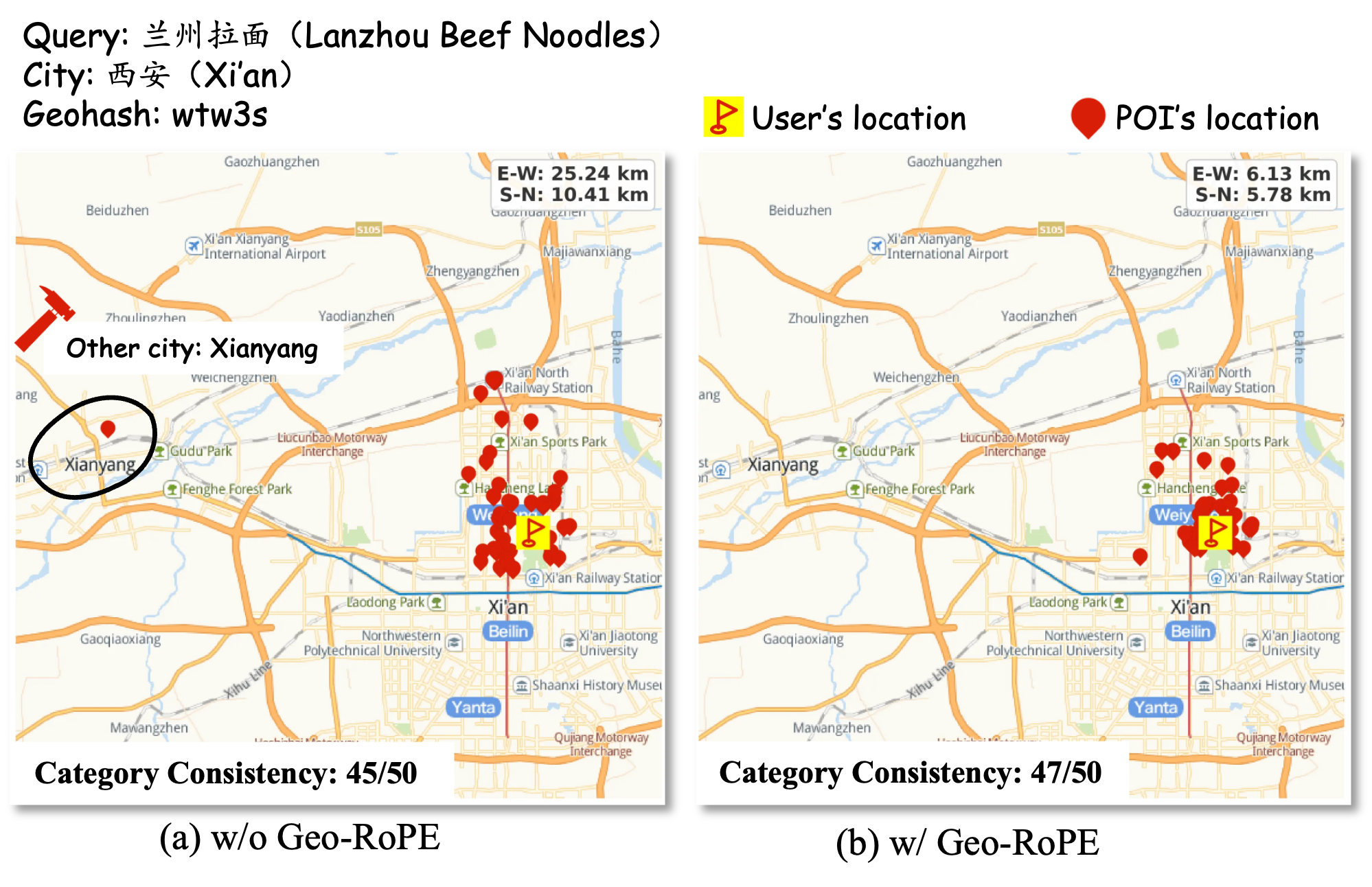}
    \vspace{-10pt}
    \caption{Visualization of recommended POI distributions with and without Geo-RoPE. E‑W indicates the maximum distance between POIs in the east–west direction; S‑N indicates the maximum distance between POIs in the south–north direction.}
    \vspace{-0pt}
    \label{fig:placeholder}
\end{figure}

\subsubsection{Visualization of recommendation results}
Fig. \ref{fig:placeholder} presents a visual comparison of the spatial distrbution of recommended POIs with and without Geo-RoPE. In the Fig. \ref{fig:placeholder}(a), the recommended POIs are spread over a much wider area, with some locations being as far as 30 kilometers from the user. Notably, recommendations in the other city (XianYang) are included, suggesting a lack of spatial precision in the model. In contrast, when Geo-RoPE is introduced, the recommended POIs are concentrated much closer to the user, typically within a 6-kilometer radius. This demonstrates the model's improved ability to capture fine-grained spatial relationships, thereby enhancing the relevance of recommendations. Furthermore, we compare category consistency under both configurations. The Geo-RoPE-enabled model achieves higher category consistency, indicating not only more spatially accurate but also more reliable recommendations in terms of matching user preferences. Additional case studies can be found in Appendix C.

\section{Conclusion}
In this work, we propose Pro-GEO, a proximity-aware geo-codebook framework for generative local service recommendation. Pro-GEO addresses the challenge of geographic misalignment in generative recommender systems, where semantic tokenization often overlooks fine-grained spatial structure. To this end, we design a geo-centroid local coordinate system to capture intra-cluster spatial relationships, and employ Geo-RoPE to embed relative geographic proximity within the codebook. This geometry-aware design enables semantic and spatial signals to be jointly modeled in a balanced manner, rather than treating geographic information as an auxiliary feature. Extensive experiments on a large-scale industrial dataset demonstrate that Pro-GEO significantly reduces geographic dispersion within codebook clusters while achieving consistent gains in the recommendation accuracy. Beyond local service recommendation, Pro-GEO highlights a general principle for structured tokenization in generative recommender systems: \textbf{auxiliary constraints can be more effectively integrated through attribute-aware transformations in the representation space} instead of simple feature concatenation. This perspective opens up opportunities for incorporating other domain-specific constraints, such as \textit{price sensitivity}, \textit{delivery time}, or \textit{inventory availability}, via analogous inductive biases. In summary, Pro-GEO bridges semantic relevance and geographic feasibility within a unified codebook framework. We believe this work provides a scalable and extensible foundation for future context-aware generative recommendation systems.

\bibliographystyle{ACM-Reference-Format}
\bibliography{sample-base}

@String{Computing = "Computing" }

@ArtifactSoftware{R,
    title = {R: A Language and Environment for Statistical Computing},
    author = {{R Core Team}},
    organization = {R Foundation for Statistical Computing},
    address = {Vienna, Austria},
    year = {2019},
    url = {https://www.R-project.org/},
}

@article{jiang2025plug,
  title={A Plug-and-Play Spatially-Constrained Representation Enhancement Framework for Local-Life Recommendation},
  author={Jiang, Hao and Wang, Guoquan and Yu, Sheng and Zeng, Yang and Zeng, Wencong and Zhou, Guorui},
  journal={arXiv preprint arXiv:2511.12947},
  year={2025}
}

@article{chen2025onesearch,
  title={Onesearch: A preliminary exploration of the unified end-to-end generative framework for e-commerce search},
  author={Chen, Ben and Guo, Xian and Wang, Siyuan and Liang, Zihan and Lv, Yue and Ma, Yufei and Xiao, Xinlong and Xue, Bowen and Zhang, Xuxin and Yang, Ying and others},
  journal={arXiv preprint arXiv:2509.03236},
  year={2025}
}

@article{chen2025unisearch,
  title={UniSearch: Rethinking Search System with a Unified Generative Architecture},
  author={Chen, Jiahui and Jiang, Xiaoze and Wang, Zhibo and Zhu, Quanzhi and Zhao, Junyao and Hu, Feng and Pan, Kang and Xie, Ao and Pei, Maohua and Qin, Zhiheng and others},
  journal={arXiv preprint arXiv:2509.06887},
  year={2025}
}

@inproceedings{hong2025eager,
  title={EAGER-LLM: Enhancing Large Language Models as Recommenders through Exogenous Behavior-Semantic Integration},
  author={Hong, Minjie and Xia, Yan and Wang, Zehan and Zhu, Jieming and Wang, Ye and Cai, Sihang and Yang, Xiaoda and Dai, Quanyu and Dong, Zhenhua and Zhang, Zhimeng and others},
  booktitle={Proceedings of the ACM on Web Conference 2025},
  pages={2754--2762},
  year={2025}
}

@inproceedings{hou2025generating,
  title={Generating long semantic ids in parallel for recommendation},
  author={Hou, Yupeng and Li, Jiacheng and Shin, Ashley and Jeon, Jinsung and Santhanam, Abhishek and Shao, Wei and Hassani, Kaveh and Yao, Ning and McAuley, Julian},
  booktitle={Proceedings of the 31st ACM SIGKDD Conference on Knowledge Discovery and Data Mining V. 2},
  pages={956--966},
  year={2025}
}

@article{van2017neural,
  title={Neural discrete representation learning},
  author={Van Den Oord, Aaron and Vinyals, Oriol and others},
  journal={Advances in neural information processing systems},
  volume={30},
  year={2017}
}

@article{rajput2023recommender,
  title={Recommender systems with generative retrieval},
  author={Rajput, Shashank and Mehta, Nikhil and Singh, Anima and Hulikal Keshavan, Raghunandan and Vu, Trung and Heldt, Lukasz and Hong, Lichan and Tay, Yi and Tran, Vinh and Samost, Jonah and others},
  journal={Advances in Neural Information Processing Systems},
  volume={36},
  pages={10299--10315},
  year={2023}
}

@inproceedings{zheng2024adapting,
  title={Adapting large language models by integrating collaborative semantics for recommendation},
  author={Zheng, Bowen and Hou, Yupeng and Lu, Hongyu and Chen, Yu and Zhao, Wayne Xin and Chen, Ming and Wen, Ji-Rong},
  booktitle={2024 IEEE 40th International Conference on Data Engineering (ICDE)},
  pages={1435--1448},
  year={2024},
  organization={IEEE}
}

@article{wei2025oneloc,
  title={Oneloc: Geo-aware generative recommender systems for local life service},
  author={Wei, Zhipeng and Cai, Kuo and She, Junda and Chen, Jie and Chen, Minghao and Zeng, Yang and Luo, Qiang and Zeng, Wencong and Tang, Ruiming and Gai, Kun and others},
  journal={arXiv preprint arXiv:2508.14646},
  year={2025}
}

@inproceedings{wang2025generative,
  title={Generative next poi recommendation with semantic id},
  author={Wang, Dongsheng and Huang, Yuxi and Gao, Shen and Wang, Yifan and Huang, Chengrui and Shang, Shuo},
  booktitle={Proceedings of the 31st ACM SIGKDD Conference on Knowledge Discovery and Data Mining V. 2},
  pages={2904--2914},
  year={2025}
}

@article{jiang2025llm,
  title={Llm-aligned geographic item tokenization for local-life recommendation},
  author={Jiang, Hao and Wang, Guoquan and Zhou, Donglin and Yu, Sheng and Zeng, Yang and Zeng, Wencong and Gai, Kun and Zhou, Guorui},
  journal={arXiv preprint arXiv:2511.14221},
  year={2025}
}

@article{su2024roformer,
  title={Roformer: Enhanced transformer with rotary position embedding},
  author={Su, Jianlin and Ahmed, Murtadha and Lu, Yu and Pan, Shengfeng and Bo, Wen and Liu, Yunfeng},
  journal={Neurocomputing},
  volume={568},
  pages={127063},
  year={2024},
  publisher={Elsevier}
}

@article{zhou2025hymirec,
  title={HyMiRec: A Hybrid Multi-interest Learning Framework for LLM-based Sequential Recommendation},
  author={Zhou, Jingyi and Chen, Cheng and Zuo, Kai and Xu, Manjie and Fu, Zhendong and Chen, Yibo and Tang, Xu and Hu, Yao},
  journal={arXiv preprint arXiv:2510.13738},
  year={2025}
}

@inproceedings{lian2014geomf,
  title={GeoMF: joint geographical modeling and matrix factorization for point-of-interest recommendation},
  author={Lian, Defu and Zhao, Cong and Xie, Xing and Sun, Guangzhong and Chen, Enhong and Rui, Yong},
  booktitle={Proceedings of the 20th ACM SIGKDD international conference on Knowledge discovery and data mining},
  pages={831--840},
  year={2014}
}

@inproceedings{lim2020stp,
  title={STP-UDGAT: Spatial-temporal-preference user dimensional graph attention network for next POI recommendation},
  author={Lim, Nicholas and Hooi, Bryan and Ng, See-Kiong and Wang, Xueou and Goh, Yong Liang and Weng, Renrong and Varadarajan, Jagannadan},
  booktitle={Proceedings of the 29th ACM international conference on information \& knowledge management},
  pages={845--854},
  year={2020}
}

@article{deng2025onerec,
  title={Onerec: Unifying retrieve and rank with generative recommender and iterative preference alignment},
  author={Deng, Jiaxin and Wang, Shiyao and Cai, Kuo and Ren, Lejian and Hu, Qigen and Ding, Weifeng and Luo, Qiang and Zhou, Guorui},
  journal={arXiv preprint arXiv:2502.18965},
  year={2025}
}

@inproceedings{wang2025progressive,
  title={Progressive Semantic Residual Quantization for Multimodal-Joint Interest Modeling in Music Recommendation},
  author={Wang, Shijia and Ouyang, Tianpei and Xiao, Qiang and Wang, Dongjing and Ren, Yintao and Xu, Songpei and Guo, Da and Luo, Chuanjiang},
  booktitle={Proceedings of the 34th ACM International Conference on Information and Knowledge Management},
  pages={6119--6127},
  year={2025}
}

@inproceedings{geng2022recommendation,
  title={Recommendation as language processing (rlp): A unified pretrain, personalized prompt \& predict paradigm (p5)},
  author={Geng, Shijie and Liu, Shuchang and Fu, Zuohui and Ge, Yingqiang and Zhang, Yongfeng},
  booktitle={Proceedings of the 16th ACM conference on recommender systems},
  pages={299--315},
  year={2022}
}

@article{cui2022m6,
  title={M6-rec: Generative pretrained language models are open-ended recommender systems},
  author={Cui, Zeyu and Ma, Jianxin and Zhou, Chang and Zhou, Jingren and Yang, Hongxia},
  journal={arXiv preprint arXiv:2205.08084},
  year={2022}
}

@article{zhang2025gpr,
  title={GPR: Towards a Generative Pre-trained One-Model Paradigm for Large-Scale Advertising Recommendation},
  author={Zhang, Jun and Li, Yi and Liu, Yue and Wang, Changping and Wang, Yuan and Xiong, Yuling and Liu, Xun and Wu, Haiyang and Li, Qian and Zhang, Enming and others},
  journal={arXiv preprint arXiv:2511.10138},
  year={2025}
}

@article{zhou2025onerec,
  title={Onerec-v2 technical report},
  author={Zhou, Guorui and Hu, Hengrui and Cheng, Hongtao and Wang, Huanjie and Deng, Jiaxin and Zhang, Jinghao and Cai, Kuo and Ren, Lejian and Ren, Lu and Yu, Liao and others},
  journal={arXiv preprint arXiv:2508.20900},
  year={2025}
}

@article{liu2025onerec,
  title={Onerec-think: In-text reasoning for generative recommendation},
  author={Liu, Zhanyu and Wang, Shiyao and Wang, Xingmei and Zhang, Rongzhou and Deng, Jiaxin and Bao, Honghui and Zhang, Jinghao and Li, Wuchao and Zheng, Pengfei and Wu, Xiangyu and others},
  journal={arXiv preprint arXiv:2510.11639},
  year={2025}
}

@article{yang2025sparse,
  title={Sparse meets dense: Unified generative recommendations with cascaded sparse-dense representations},
  author={Yang, Yuhao and Ji, Zhi and Li, Zhaopeng and Li, Yi and Mo, Zhonglin and Ding, Yue and Chen, Kai and Zhang, Zijian and Li, Jie and Li, Shuanglong and others},
  journal={arXiv preprint arXiv:2503.02453},
  year={2025}
}

@article{wang2025gflowgr,
  title={GFlowGR: Fine-tuning Generative Recommendation Frameworks with Generative Flow Networks},
  author={Wang, Yejing and Zhou, Shengyu and Lu, Jinyu and Liu, Qidong and Li, Xinhang and Zhang, Wenlin and Li, Feng and Wang, Pengjie and Xu, Jian and Zheng, Bo and others},
  journal={arXiv preprint arXiv:2506.16114},
  year={2025}
}

@article{zhang2025slow,
  title={Slow Thinking for Sequential Recommendation},
  author={Zhang, Junjie and Zhang, Beichen and Sun, Wenqi and Lu, Hongyu and Zhao, Wayne Xin and Chen, Yu and Wen, Ji-Rong},
  journal={arXiv preprint arXiv:2504.09627},
  year={2025}
}

@article{han2026feature,
  title={Feature-Indexed Federated Recommendation with Residual-Quantized Codebooks},
  author={Han, Mingzhe and Liu, Jiahao and Li, Dongsheng and Gu, Hansu and Zhang, Peng and Gu, Ning and Lu, Tun},
  journal={arXiv preprint arXiv:2601.18570},
  year={2026}
}

@article{zhang2025qwen3,
  title={Qwen3 Embedding: Advancing Text Embedding and Reranking Through Foundation Models},
  author={Zhang, Yanzhao and Li, Mingxin and Long, Dingkun and Zhang, Xin and Lin, Huan and Yang, Baosong and Xie, Pengjun and Yang, An and Liu, Dayiheng and Lin, Junyang and others},
  journal={arXiv preprint arXiv:2506.05176},
  year={2025}
}

@article{yang2025qwen2,
  title={Qwen2. 5-1m technical report},
  author={Yang, An and Yu, Bowen and Li, Chengyuan and Liu, Dayiheng and Huang, Fei and Huang, Haoyan and Jiang, Jiandong and Tu, Jianhong and Zhang, Jianwei and Zhou, Jingren and others},
  journal={arXiv preprint arXiv:2501.15383},
  year={2025}
}

@inproceedings{wang2025fim,
  title={FIM: Frequency-aware multi-view interest modeling for local-life service recommendation},
  author={Wang, Guoquan and Luo, Qiang and Hu, Weisong and Yao, Pengfei and Zeng, Wencong and Zhou, Guorui and Gai, Kun},
  booktitle={Proceedings of the 48th International ACM SIGIR Conference on Research and Development in Information Retrieval},
  pages={1748--1757},
  year={2025}
}

@article{hu2025dynamic,
  title={Dynamic Forgetting and Spatio-Temporal Periodic Interest Modeling for Local-Life Service Recommendation},
  author={Hu, Zhaoyu and Wang, Jianyang and Guo, Hao and Tian, Yuan and Xue, Erpeng and Qi, Xianyang and Lin, Hongxiang and Wang, Lei and Chen, Sheng},
  journal={arXiv preprint arXiv:2508.02451},
  year={2025}
}

@article{li2025survey,
  title={A survey of generative recommendation from a tri-decoupled perspective: Tokenization, architecture, and optimization},
  author={Li, Xiaopeng and Chen, Bo and She, Junda and Cao, Shiteng and Wang, You and Jia, Qinlin and He, Haiying and Zhou, Zheli and Liu, Zhao and Liu, Ji and others},
  year={2025},
  publisher={Preprints}
}

@article{yang2014modeling,
  title={Modeling user activity preference by leveraging user spatial temporal characteristics in LBSNs},
  author={Yang, Dingqi and Zhang, Daqing and Zheng, Vincent W and Yu, Zhiyong},
  journal={IEEE Transactions on Systems, Man, and Cybernetics: Systems},
  volume={45},
  number={1},
  pages={129--142},
  year={2014},
  publisher={IEEE}
}

\appendix

\section{Proof of Equation~\eqref{eq:11}}

\subsection{Notations and Preliminaries}

For ease of presentation, we introduce the following definitions.

\textbf{Definition 1 (Mirror-Duplicated Vector).}  
A vector $\mathbf{v} \in \mathbb{R}^{2d}$ is called \emph{mirror-duplicated} if there exists $\mathbf{v}_1 \in \mathbb{R}^d$ such that
\begin{equation}
    \mathbf{v} = \begin{bmatrix} \mathbf{v}_1 \\ \mathbf{v}_1 \end{bmatrix}.
\end{equation}


\textbf{Definition 2 (Geo-RoPE Transformation).}  Let $d = 2m$. Denote the standard $2 \times 2$ rotation matrix by
\begin{equation}
    R(\theta) = \begin{bmatrix} \cos\theta & -\sin\theta \\ \sin\theta & \cos\theta \end{bmatrix}.
\end{equation}
Define the block-diagonal rotation matrix as
\begin{equation}
    \mathcal{R}(\theta) = \mathrm{diag}(R(\theta), \ldots, R(\theta)) \in \mathbb{R}^{d \times d},
\end{equation}
which contains $m$ diagonal blocks.

For vector $\mathbf{v}_1$, the \textbf{\emph{Geo-RoPE transformation}} $\mathcal{T}_\theta: \mathbb{R}^{d} \to \mathbb{R}^{2d}$ is defined as
\begin{equation}
    \mathcal{T}_\theta(\mathbf{v_1}) = \begin{bmatrix} \mathcal{R}(\theta) \mathbf{v}_1 \\ \mathcal{R}(-\theta) \mathbf{v}_1 \end{bmatrix}.
\end{equation}

\textbf{Definition 3 (Cosine Distance).}  
For nonzero vectors $\mathbf{u}, \mathbf{v} \in \mathbb{R}^n $, the cosine distance is defined as
\begin{equation}
    D_{\cos}(\mathbf{u}, \mathbf{v}) = 1 - \frac{\langle \mathbf{u}, \mathbf{v} \rangle}{\|\mathbf{u}\| \|\mathbf{v}\|}.
\end{equation}

\subsection{Key Lemma}

\begin{lemma}\label{lemma:inner-product}
Let $\mathbf{r}_1 = \begin{bmatrix} \overline{\mathbf{r}}_1 \\ \overline{\mathbf{r}}_1 \end{bmatrix}$ and $\mathbf{r}_2 = \begin{bmatrix} \overline{\mathbf{r}}_2 \\ \overline{\mathbf{r}}_2 \end{bmatrix}$ be two mirror-duplicated vectors. After applying the Geo-RoPE transformations $\mathbf{r}_1' = \mathcal{T}_{\theta_1}(\overline{\mathbf{r}}_1)$ and $\mathbf{r}_2' = \mathcal{T}_{\theta_2}(\overline{\mathbf{r}}_2)$, the inner product satisfies
\begin{equation}
    \langle \mathbf{r}_1', \mathbf{r}_2' \rangle = \cos(\Delta\theta) \cdot \langle \mathbf{r}_1, \mathbf{r}_2 \rangle,
\end{equation}
where $\Delta\theta = \theta_1 - \theta_2$.
\end{lemma}

\begin{proof}
Decompose $\overline{\mathbf{r}}_1$ and $\overline{\mathbf{r}}_2$ into $m$ pairs of 2-dimensional sub-vectors:
\begin{equation}
    \overline{\mathbf{r}}_1 = \begin{bmatrix} 
    \overline{\mathbf{r}}_1^{(1)} \\ 
    \vdots \\ 
    \overline{\mathbf{r}}_1^{(m)} 
    \end{bmatrix}, \quad
    \overline{\mathbf{r}}_2 = \begin{bmatrix} 
    \overline{\mathbf{r}}_2^{(1)} \\ 
    \vdots \\ 
    \overline{\mathbf{r}}_2^{(m)} 
    \end{bmatrix},
\end{equation}
where
\begin{equation}
    \overline{\mathbf{r}}_j^{(i)} = \begin{bmatrix} 
    a_j^{(i)} \\ 
    b_j^{(i)} 
    \end{bmatrix}, \quad
    j \in \{1, 2\}, \; i = 1, \ldots, m, \quad
    m = \frac{\dim(\overline{\mathbf{r}}_1)}{2} = \frac{\dim(\overline{\mathbf{r}}_2)}{2}.
\end{equation}

For each pair of vectors $\overline{\mathbf{r}}_1^{(i)}$ and $\overline{\mathbf{r}}_2^{(i)}$, define:
\begin{align}
    d_i &= \langle \overline{\mathbf{r}}_1^{(i)}, \overline{\mathbf{r}}_2^{(i)} \rangle = a_1^{(i)} a_2^{(i)} + b_2^{(i)} b_2^{(i)}, \label{eq:dot-component}\\
    c_i &= a_1^{(i)} b_2^{(i)} - a_2^{(i)}b_1^{(i)}. \label{eq:cross-component}
\end{align}

By the mirror-duplication property, the original inner product is
\begin{equation}
    \langle \mathbf{r}_1, \mathbf{r}_2 \rangle = 2\langle \overline{\mathbf{r}}_1, \overline{\mathbf{r}}_2 \rangle = 2\sum_{i=1}^{m} d_i.
\end{equation}

Using the property of rotation matrices $\langle R(\alpha)\mathbf{u}, R(\beta)\mathbf{v} \rangle = \mathbf{u}^T R(\beta-\alpha) \mathbf{v}$ (which follows from $R(\alpha)^T = R(-\alpha)$), we have:

\textbf{Forward rotation part} (applying $R(\theta_1)$ and $R(\theta_2)$ to the $i$-th pair):
\begin{equation}
    \langle R(\theta_1)\overline{\mathbf{r}}_1^{(i)}, R(\theta_2)\overline{\mathbf{r}}_2^{(i)} \rangle = (\overline{\mathbf{r}}_1^{(i)})^T R(\theta_2-\theta_1) \overline{\mathbf{r}}_2^{(i)} = (\overline{\mathbf{r}}_1^{(i)})^T R(-\Delta\theta) \overline{\mathbf{r}}_2^{(i)}.
\end{equation}

Expanding $R(-\Delta\theta) = \begin{bmatrix} \cos\Delta\theta & \sin\Delta\theta \\ -\sin\Delta\theta & \cos\Delta\theta \end{bmatrix}$, we compute:
\begin{equation}
    R(-\Delta\theta) \overline{\mathbf{r}}_2^{(i)} = \begin{bmatrix} a_2^{(i)}\cos\Delta\theta + b_2^{(i)}\sin\Delta\theta \\ -a_2^{(i)}\sin\Delta\theta + b_2^{(i)}\cos\Delta\theta \end{bmatrix}.
\end{equation}

Taking the inner product with $\overline{\mathbf{r}}_1^{(i)}$:
\begin{align}
    &(\overline{\mathbf{r}}_1^{(i)})^T R(-\Delta\theta) \overline{\mathbf{r}}_2^{(i)} \notag \\
    &= a_1^{(i)}[a_2^{(i)}\cos\Delta\theta + b_2^{(i)}\sin\Delta\theta] \notag + b_1^{(i)}[-a_2^{(i)}\sin\Delta\theta + b_2^{(i)}\cos\Delta\theta] \notag \\
    &= (a_1^{(i)}a_2^{(i)} + b_1^{(i)}b_2^{(i)})\cos\Delta\theta + (a_1^{(i)}b_2^{(i)} - a_2^{(i)}b_1^{(i)})\sin\Delta\theta \notag \\
    &= d_i\cos\Delta\theta + c_i\sin\Delta\theta.
\end{align}

\textbf{Reverse rotation part} (applying $R(-\theta_1)$ and $R(-\theta_2)$ to the $i$-th pair):  

By a completely analogous calculation, replacing $R(-\Delta\theta)$ with $R(\Delta\theta)$, we obtain:
\begin{equation}
    \langle R(-\theta_1)\overline{\mathbf{r}}_1^{(i)}, R(-\theta_2)\overline{\mathbf{r}}_2^{(i)} \rangle = d_i\cos\Delta\theta - c_i\sin\Delta\theta.
\end{equation}

Now summing over all $m$ pairs. Since $\mathbf{r}_1' = \begin{bmatrix} \mathcal{R}(\theta_1)\overline{\mathbf{r}}_1 \\ \mathcal{R}(-\theta_1)\overline{\mathbf{r}}_1 \end{bmatrix}$ and $\mathbf{r}_2' = \begin{bmatrix} \mathcal{R}(\theta_2)\overline{\mathbf{r}}_2 \\ \mathcal{R}(-\theta_2)\overline{\mathbf{r}}_2 \end{bmatrix}$, the total inner product is:
\begin{align}
\label{eq:29}
    \langle \mathbf{r}_1', \mathbf{r}_2' \rangle 
    &= \sum_{i=1}^{m} \langle R(\theta_1)\overline{\mathbf{r}}_1^{(i)}, R(\theta_2)\overline{\mathbf{r}}_2^{(i)} \rangle + \sum_{i=1}^{m} \langle R(-\theta_1)\overline{\mathbf{r}}_1^{(i)}, R(-\theta_2)\overline{\mathbf{r}}_2^{(i)} \rangle \notag \\
    &= \sum_{i=1}^{m} [d_i\cos\Delta\theta + c_i\sin\Delta\theta] + \sum_{i=1}^{m} [d_i\cos\Delta\theta - c_i\sin\Delta\theta] \notag \\
    &= 2\cos\Delta\theta \sum_{i=1}^{m} d_i + \sin\Delta\theta \underbrace{\left[\sum_{i=1}^{m} c_i - \sum_{i=1}^{m} c_i\right]}_{=0} \notag \\
    &= 2\cos\Delta\theta \sum_{i=1}^{m} d_i = \cos\Delta\theta \cdot \langle \mathbf{r}_1, \mathbf{r}_2 \rangle.
\end{align}

As demonstrated by the Eq.\ref{eq:29}, the symmetric forward-reverse rotation structure allows the transformed vector distance to be determined exclusively by their difference of the relative rotational angle, independent of their absolute rotation angles. This structure preserves the original semantic relationships between vectors while incorporating additional information, such as the geographic feature utilized in this work.  In Table \ref{tab:comparison}, experimental results demonstrate that the codebook designed with the forward-reverse rotation structure exhibits significant advantages across all prediction metrics.
\end{proof}
\subsection{Proof of Equation~\eqref{eq:11}}

Since $\mathcal{T}_\theta$ is orthogonal, it preserves vector norms:
\begin{equation}
    \|\mathbf{r}_i'\| = \|\mathbf{r}_i\|, \quad \|\mathbf{r}_j'\| = \|\mathbf{r}_j\|.
\end{equation}


According to \textbf{Lemma A.1}, the cosine similarity after transformation is
\begin{equation}
    \cos\alpha' = \frac{\langle \mathbf{r}_1', \mathbf{r}_2' \rangle}{\|\mathbf{r}_1'\| \|\mathbf{r}_2'\|} = \frac{\cos\Delta\theta \cdot \langle \mathbf{r_1}, \mathbf{r_2} \rangle}{\|\mathbf{r_1}\| \|\mathbf{r_2}\|} = \cos\Delta\theta \cdot \cos\alpha,
    \label{eq:cos-after}
\end{equation}
where $\cos\alpha = \frac{\langle \mathbf{r}_i, \mathbf{r}_j \rangle}{\|\mathbf{r}_i\|\,\|\mathbf{r}_j\|}$.

Therefore, the change in cosine distance is
\begin{align}
\Delta D_{\cos} 
&= D_{\cos}(\mathbf{r}_i', \mathbf{r}_j') - D_{\cos}(\mathbf{r}_i, \mathbf{r}_j) \notag \\
&= [1 - \cos\alpha'] - [1 - \cos\alpha] \notag \\
&= \cos\alpha - \cos\alpha' \notag \\
&= \cos\alpha[1 - \cos(\Delta\theta)].
\label{eq:dc-change}
\end{align}
Using the trigonometric identity $1 - \cos\Delta\theta = 2\sin^2(\Delta\theta/2)$, we obtain
\begin{equation}
\Delta D_{\cos} = 2\cos\alpha \cdot \sin^2\left(\frac{\Delta\theta}{2}\right) = 2 \cdot \frac{\langle \mathbf{r}_i, \mathbf{r}_j \rangle}{\|\mathbf{r}_i\|\,\|\mathbf{r}_j\|} \cdot \sin^2\left(\frac{\Delta\theta}{2}\right).
\label{eq:final}
\end{equation}

Thus, 
\begin{center}
\boxed{
\Delta D_{\cos} = 2 \cdot \frac{\langle \mathbf{r}_i, \mathbf{r}_j \rangle}{\|\mathbf{r}_i\|\,\|\mathbf{r}_j\|} \cdot \sin^2\left(\frac{\Delta\theta}{2}\right)
}
\end{center}

\subsection{Conclusion}

The above proof rigorously establishes the change in cosine distance under the Geo-RoPE transformation. The key insight is:

\begin{center}
\boxed{
\begin{aligned}
    &\text{Cosine distance change between two vectors:} \\[8pt]
    &\qquad \Delta D_{\cos} \propto 
        \underbrace{\sin^2\left(\frac{\Delta\theta}{2}\right)}_{\substack{\text{Positional} \\ \text{difference factor}}}
        \times 
        \underbrace{\cos\alpha}_{\substack{\text{Semantic} \\ \text{similarity factor}}}
\end{aligned}
}
\end{center}

This formula demonstrates that the geo-RoPE transformation effectively 
increases the cosine distance between vectors that are geographically 
distant, while preserving the semantic similarity between vectors 
that are semantically close.

\section{Experimental Setup}
\subsection{Dataset}
Our experiments are conducted on real-world industrial data from a leading local-life service platform in China. The dataset is sourced from this platform’s core local-life application scenario, reflecting authentic user interactions and purchase behaviors. Specifically, we collected users' transaction records within two months for model training, and reserved additional daily data for validation. Each data contains the complete user search query, user geographic location, merchant exposure sequence, and actual user interaction behavior. Moreover, we extended our experiments to multiple public POI recommendation datasets with geographic information, namely Foursquare-NYC \cite{yang2014modeling}, Foursquare-TKY \cite{yang2014modeling}, and Yelp. Foursquare-NYC and Foursquare-TKY include both merchant information and user behavioral sequences, while Yelp contains only merchant information. All three datasets were used to evaluate the clustering effect of Pro-GEO, while recommendation effectiveness was compared on NYC and TKY. Following GNPR-SID \cite{wang2025generative}, we preprocess each dataset by removing POIs with fewer than 10 interactions and users with fewer than 10 check-ins.

\subsection{Implementation Details}
\subsubsection{Industrial dataset}
For item tokenization training, we configure codebooks as $[512, 512, 512]$ for the industrial dataset. RoPE-based Rotary encoding is incorporated at the third quantization layer, with hyperparameters set to $\alpha=0.5$ and $\beta=0.5$. Text embeddings are extracted using the pre-trained Qwen3-Embedding-8B model \cite{zhang2025qwen3},  yielding a 128-dimensional representation for each item. For generative recommendation training, the Qwen 2.5-1.5B \cite{yang2025qwen2} is trained, with a maximum sequence length of 512 and a batch size of 64 to accommodate complex user-item interactions. We employ the AdamW optimizer with a learning rate of 8e-5 and a weight decay of 0.01. During inference, we select the 10 POIs whose semantic IDs are closest to the user's location as the candidate set. Beam search is performed with a width of 20 to generate the final list of recommended POIs. All experiments were conducted on 8 GPU with 80GB of memory. The prompt template is shown in Table~\ref{tab:prompt_design1}.

\begin{table}[htbp]
\centering
\caption{Prompt template for generative recommendation on the industrial dataset.}
\label{tab:prompt_design1}
\begin{tabularx}{\linewidth}{>{\raggedright\arraybackslash}X}
\toprule
\textbf{Instruction:} \\
Here is a local life service expert. Your task is to predict the list of POIs that the user might purchase based on their intentions. \\
\midrule
\textbf{Input:}   \\
Query: \textcolor{orange}{<query>}, City: \textcolor{green}{<city>} GeoHash: \textcolor{purple}{<geohash>} \\
\midrule
\textbf{Output:} \\
SID list: \textcolor{gray}{<Sid\_1>, <Sid\_2>, ..., <Sid\_N>} \\
\bottomrule
\end{tabularx}
\end{table}

\subsubsection{Pubilc dataset}

Following GNPR-SID \cite{wang2025generative}, we filtered NYC and TKY datasets. NYC has 1,083 users and 5,135 POIs, TKY has 2,293 users and 7,873 POIs, and Yelp contains 150,346 POIs. Codebook sizes are set to [32,32,32] for NYC/TKY and [64,64,64] for Yelp. All POI embeddings are obtained using pre-trained qwen3-embedding. For generative recommendation on NYC and TKY, we train Qwen2.5-1.5B \cite{yang2025qwen2} using AdamW (lr=1e-5, weight decay=0.01, batch size=32). Since user's geographic information is unavailable in these datasets, we follow TIGER~\cite{rajput2023recommender} and introduce Layer-4 hard coding to avoid codebook conflicts. All experiments were conducted on 8 GPU with 80GB of memory. The prompt template is shown in Table~\ref{tab:prompt_design2}. For brevity, we defer the experimental results on public datasets 
to Appendix~\ref{appendix:D}.
\begin{table}[htbp]
\centering
\caption{Prompt template for generative recommendation on the public dataset.}
\label{tab:prompt_design2}
\begin{tabularx}{\linewidth}{>{\raggedright\arraybackslash}X}
\toprule
\textbf{Instruction:} \\
Here is a record of a user's POI accesses, your task is based on the history to predict the POI that the user is likely to access at the specified time. \\
\midrule
\textbf{Input:}   \\
The user \textcolor{orange}{<uid>} visited: \textcolor{green}{<Sid\_1>} at \textcolor{blue}{<time\_1>}, ..., visited: \textcolor{green}{<Sid\_N>} at \textcolor{blue}{<time\_N>}. When \textcolor{blue}{<time>} user \textcolor{orange}{<uid>} is likely to visit:\\
\midrule
\textbf{Output:} \textcolor{gray}{<Sid>} \\
\bottomrule
\end{tabularx}
\end{table}

\subsection{Evaluation Metrics}
\begin{figure*}[h]
    \centering
    \includegraphics[width=\textwidth]{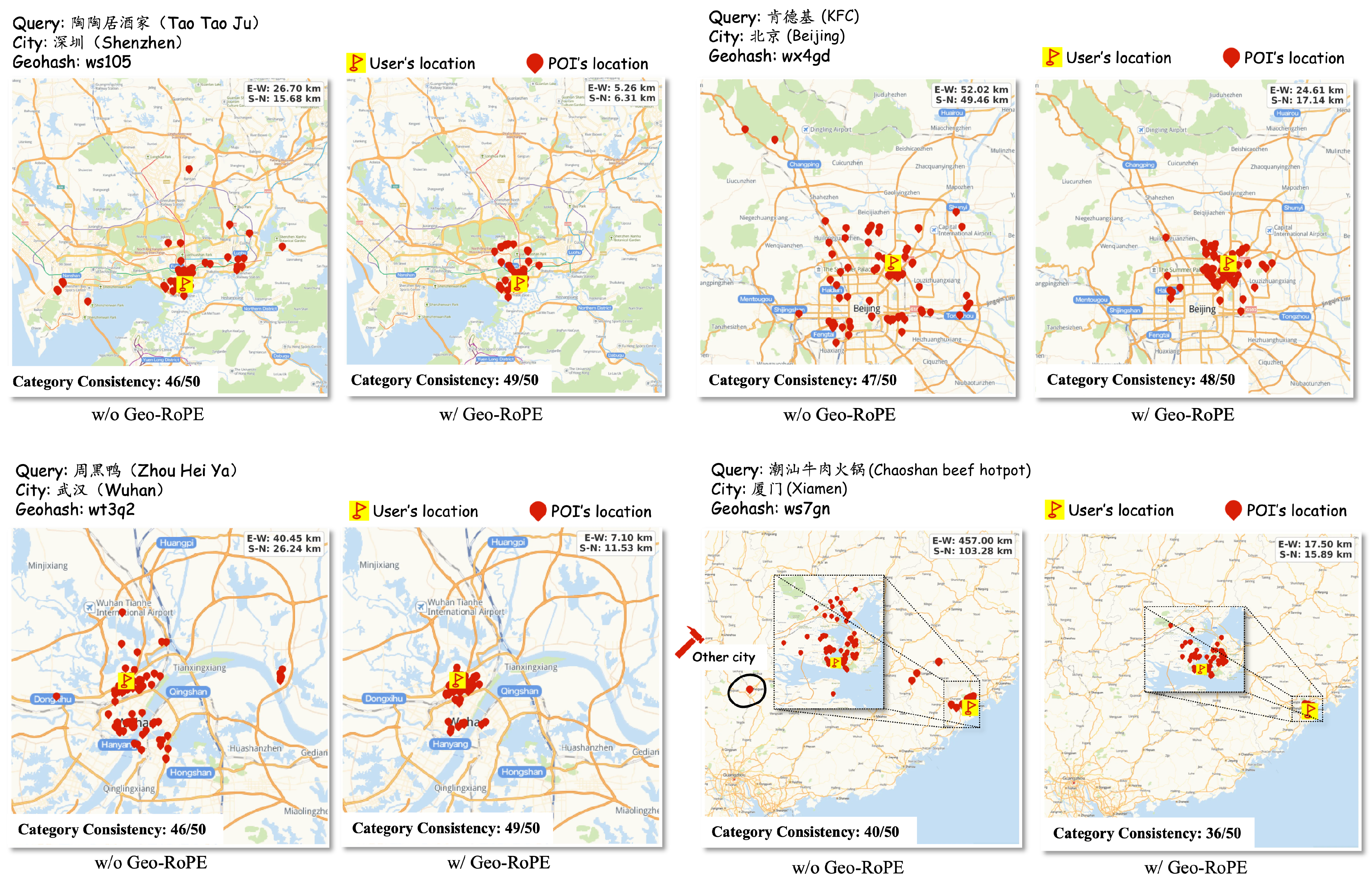}
    \caption{Additional case studies on the spatial distribution of POI recommendation results. E‑W indicates the maximum distance between POIs in the east–west direction; S‑N indicates the maximum distance between POIs in the south–north direction.}
    \label{fig:124245}
\end{figure*}
\begin{table*}[h]
\centering
\caption{Quantization and recommendation metrics comparison on public datasets.
Positive improvements are shown in \colorbox{green}{\textcolor{black}{positive}} 
and negative improvements in \colorbox{pink}{\textcolor{black}{negative}}. The unit for Avg. Dist., p90 Dist., and p95 Dist. is kilometers (km). $\uparrow$ and $\downarrow$ indicate that higher and lower values are better for the corresponding evaluation metrics, respectively.}
\label{tab:public_results}
\resizebox{\textwidth}{!}{
\begin{tabular}{lccccccccc}
\toprule
\multirow{2}{*}{\textbf{Metric}} 
& \multicolumn{3}{c}{\textbf{NYC}} 
& \multicolumn{3}{c}{\textbf{TKY}} 
& \multicolumn{3}{c}{\textbf{Yelp}} \\
\cmidrule(lr){2-4} \cmidrule(lr){5-7} \cmidrule(lr){8-10}
& RQ-Kmeans & Pro-GEO & \textit{Imp.} 
& RQ-Kmeans & Pro-GEO & \textit{Imp.}
& RQ-Kmeans & Pro-GEO & \textit{Imp.} \\
\midrule

CUR ($\uparrow$)      
& 4.51\%  & 5.25\%  & \gain{+0.74\%}    
& 5.21\%  & 5.59\%  & \gain{+0.38\%}    
& 28.93\% & 28.69\% & \drop{$-$0.24\%}  \\

ICR ($\uparrow$)      
& 52.13\% & 55.29\% & \gain{+3.16\%}    
& 49.12\% & 45.60\% & \drop{$-$3.52\%}  
& 60.95\% & 58.98\% & \drop{$-$1.97\%}  \\

Avg. Dist. ($\downarrow$) 
& 6.96    & 6.10    & \gain{$-$0.86}   
& 7.21    & 6.19    & \gain{$-$1.02}   
& 115.67  & 75.35   & \gain{$-$40.32}  \\

p90 Dist. ($\downarrow$)  
& 15.35   & 13.15   & \gain{$-$2.20}   
& 14.10   & 12.11   & \gain{$-$1.99}   
& 137.64  & 27.10   & \gain{$-$110.54} \\

p95 Dist. ($\downarrow$)  
& 18.60   & 16.62   & \gain{$-$1.98}   
& 16.65   & 14.75   & \gain{$-$1.90}   
& 614.16  & 322.43  & \gain{$-$291.73} \\

\midrule

Hit@1 ($\uparrow$)  
& 0.2308  & 0.3587  & \gain{+0.1279} 
& 0.1911  & 0.2833  & \gain{+0.0922} 
& -       & -       & -              \\

Hit@5 ($\uparrow$)  
& 0.4123  & 0.5300  & \gain{+0.1177} 
& 0.3000  & 0.4927  & \gain{+0.1927} 
& -       & -       & -              \\

Hit@10 ($\uparrow$) 
& 0.4591  & 0.5763  & \gain{+0.1172} 
& 0.3396  & 0.5583  & \gain{+0.2187} 
& -       & -       & -              \\

\bottomrule
\end{tabular}
}
\end{table*}
We employ the quantization metrics to evaluate the performance of the codebook and the recommendation metrics to evaluate the performance of the generated semantic IDs.

\textbf{Quantization metrics:} We adopt two widely recognized metrics: independent coding rate (ICR) and codebook utilization rate (CUR). ICR measures the degree to which codewords are assigned uniquely across items. CUR focuses on the distribution of codeword assignments. Beyond these standard metrics, local service platforms impose unique spatial constraints due to delivery distance requirements. Therefore, we further analyze the geographic distribution of POIs assigned to the same SID. We introduce the Average Distance from the Center point (Avg. Dist\_q), which computes the mean Euclidean distance from each POI's position to the centroid of its quantification SID, providing a measure of intra-group spatial compactness. Additionally, we report p90. Dist\_q and p95. Dist\_q, representing the maximum distance to the SID centroid within the top 90\% and 95\% of POIs, respectively. These metrics capture the spatial spread and outlier sensitivity of each SID cluster.

\textbf{Recommendation Metrics:} We employ Hit@N and NDCG@N with$ N \in {50, 100}$ to evaluate the performance of generative recommendation. Additionally, we report Hit@1 to compare the top-1 recommendation results across different methods. 

\subsection{Comparison Methods}

To evaluate the effectiveness of our proposed codebook Pro-GEO, we compare it with several state-of-the-art tokenization methods within the Qwen2.5-1.5B-based generative recommendation model. (1) RQ-VAE \cite{rajput2023recommender}, which uses a residual quantized variational autoencoder to learn hierarchical discrete representations for POIs; (2) RQ-Kmeans \cite{deng2025onerec}, which applies residual quantization with k-means clustering to compress high-dimensional item embeddings into compact token sequences; (3) Zhou et al. \cite{zhou2025hymirec}, which modifies RQ-Kmeans by replacing the Euclidean distance metric with cosine similarity and using projected residuals instead of direct subtraction; (4) RQ-OPQ \cite{chen2025onesearch}, which combines residual quantization with optimized product quantization to improve token efficiency and retrieval accuracy in large-scale recommendation systems.

This comparative perspective allows us to systematically analyze how different codebook design choices impact the performance of local service recommendation.



\section{More case studies for reasoning results}

Figure \ref{fig:124245} offers additional case studies on the spatial distribution and category consistency of POI recommendation results, comparing settings with and without Geo-RoPE across different cities. Across all cases, the results consistently demonstrate the advantages brought by Geo-RoPE. Without Geo-RoPE (left column in each group), the recommended POIs are often dispersed over a broader geographic region, occasionally spanning distances of 20–36 km and even extending into neighboring cities. This large spatial spread can negatively affect both the practicality and relevance of recommendation results, as some suggested POIs are situated far away from the user’s actual location. In contrast, in the presence of Geo-RoPE (right column in each group), recommended POIs are markedly concentrated within a much smaller area, generally within 6–10 km of the user. The reduced spatial span directly demonstrates the model’s improved ability to understand and utilize fine-grained geographic contexts in recommendation generation. Furthermore, the category consistency metric shows that Geo-RoPE consistently leads to more reliable recommendations that better match the user’s interests. This highlights that Geo-RoPE not only enhances spatial relevance but also maintains or improves semantic reliability in diverse scenarios. Overall, these visualizations provide strong empirical evidence for the effectiveness of Geo-RoPE in generating POI recommendations that are both geographically and semantically consistent.

\section{Public dataset results}
\label{appendix:D}
Table~\ref{tab:public_results} presents the quantization and recommendation 
results on three public datasets. We analyze the results from two perspectives.

\textbf{Quantization Metrics.}
Pro-GEO consistently outperforms RQ-KMeans across all geographic clustering 
metrics. On Yelp, the geographic cluster distance drops substantially 
(Avg. Dist. $\downarrow$40.32km; p90 Dist. $\downarrow$110.54km; 
p95 Dist. $\downarrow$291.73km), demonstrating Pro-GEO's strong capability 
to enhance geo-awareness at large scale. On NYC, Pro-GEO also achieves 
notable gains in both CUR (+0.74\%) and ICR (+3.16\%). 
Although marginal declines in ICR are observed on TKY and Yelp, 
the codebook remains well-utilized without collapse, indicating that 
Pro-GEO effectively improves geographic consistency while preserving 
semantic discriminability.

\textbf{Recommendation Metrics.}
On NYC and TKY, integrating Pro-GEO into Qwen2.5-1.5B yields consistent 
improvements across all metrics. Specifically, Hit@1 improves by 12.79\% 
on NYC and 9.22\% on TKY, with substantial gains in Hit@5 and Hit@10 
as well. These results demonstrate that the improved geographic consistency 
in POI tokenization directly translates to better downstream recommendation 
performance.

In summary, Pro-GEO exhibits consistently strong performance in both 
geographic clustering and recommendation across datasets of varying scales 
and domains, fully demonstrating its effectiveness and generalizability.








\end{document}